\documentclass[11pt,reqno]{article}
\usepackage{amssymb,amscd,amsmath,amsthm,color}

\newtheorem{thm}{Theorem}[section]
\newtheorem{prop}[thm]{Proposition}
\newtheorem{lemma}[thm]{Lemma}
\newtheorem{cor}[thm]{Corollary}
\newtheorem{definition}[thm]{Definition}
\newtheorem{remark}[thm]{Remark}
\newtheorem{example}[thm]{Example}

\numberwithin{equation}{section}

\def\bR{\mathbb{R}}
\def\bC{\mathbb{C}}
\def\bN{\mathbb{N}}
\def\bM{\mathbb{M}}
\def\cH{\mathcal{H}}
\def\cB{\mathcal{B}}

\def\cD{\mathcal{D}}

\def\Tr{\mathrm{Tr}\,}

\def\eps{\varepsilon}
\def\<{\langle}
\def\>{\rangle}
\def\ffi{\varphi}
\def\tr{\mathrm{tr}}

\def\Im{\mathrm{Im}\,}

\def\cA{\mathcal{A}}
\def\cM{\mathcal{M}}
\def\cX{\mathcal{X}}
\def\BS{\mathrm{BS}}
\def\meas{\mathrm{meas}}
\def\pr{\mathrm{pr}}

\begin{document}
\allowdisplaybreaks

\centerline{\LARGE Quantum $f$-divergences in von Neumann algebras II.}
\medskip
\centerline{\LARGE Maximal $f$-divergences}
\bigskip
\bigskip
\centerline{\Large Fumio Hiai\footnote{{\it E-mail address:} hiai.fumio@gmail.com}}

\begin{center}
$^1$\,Graduate School of Information Sciences, \\
Tohoku University, Aoba-ku, Sendai 980-8579, Japan
\end{center}

\medskip
\begin{abstract}

As a continuation of the paper \cite{Hi1} on standard $f$-divergences, we make a systematic
study of maximal $f$-divergences in general von Neumann algebras. For maximal $f$-divergences,
apart from their definition based on Haagerup's $L^1$-space, we present the general integral
expression and the variational expression in terms of reverse tests. From these definition and
expressions we prove important properties of maximal $f$-divergences, for instance, the
monotonicity inequality, the joint convexity, the lower semicontinuity, and the martingale
convergence. The inequality between the standard and the maximal $f$-divergences is also given.


\end{abstract}

\section{Introduction}

Quantum divergences play a significant role in quantum information theory. We have mainly two
different kinds of quantum $f$-divergences parametrized by convex (often assumed operator
convex) functions $f$ on $(0,+\infty)$. The one is the \emph{standard $f$-divergence}
$S_f(\rho\|\sigma)$ and the other is the \emph{maximal $f$-divergence}
$\widehat S_f(\rho\|\sigma)$ (also denoted by $S_f^{\max}(\rho\|\sigma)$). When specialized to
the finite-dimensional (or the matrix) case, those $f$-divergences are defined for positive
operators $\rho,\sigma$ (for simplicity, assumed invertible) as follows: 
\begin{align}
\widehat S_f(\rho\|\sigma)&:=\Tr\sigma^{1/2}f(\sigma^{-1/2}\rho\sigma^{-1/2})\sigma^{1/2},
\label{F-1.1}\\
S_f(\rho\|\sigma)&:=\Tr\sigma^{1/2}f(L_\rho R_{\sigma^{-1}})(\sigma^{1/2}),
\label{F-1.2}
\end{align}
where $L_\rho$ and $R_{\sigma^{-1}}$ are the left multiplication by $\rho$ and the right
multiplication by $\sigma^{-1}$. (We extend the above definitions to general positive operators
$\rho,\sigma$ by converegence, see \cite{HMPB,HM}.) The standard $f$-divergence
$S_f(\rho\|\sigma)$ was formerly introduced and studied by Petz \cite{Pe,Pe0} in a more general
form in the von Neumann algebra setting under the name \emph{quasi-entropy} (which first
appeared in \cite{Ko0}). A typical and the most important example is the \emph{relative entropy}
$D(\rho\|\sigma)$ that is $S_f(\rho\|\sigma)$ when $f(t)=t\log t$, introduced first by Umegaki
\cite{Um} in semifinite von Neumann algebras, and extended to general von Neumann algebras by
Araki \cite{Ar5,Ar2} based on the relative modular operators. (Note that
$L_\rho R_{\sigma^{-1}}$ in \eqref{F-1.2} is the form of relative modular operator in the
finite-dimensional case.) On the other hand, the maximal $f$-divergence
$\widehat S_f(\rho\|\sigma)$ for matrices were studied by Matsumoto \cite{Ma}. A special
example of $\widehat S_f(\rho\|\sigma)$ when $f(t)=t\log t$ is another version of the relative
entropy introduced by Belavkin and Staszewski \cite{BS}, denoted by $D_\BS(\rho\|\sigma)$. It
is also worth noting that the form $\sigma^{1/2}f(\sigma^{-1/2}\rho\sigma^{-1/2})\sigma^{1/2}$
in \eqref{F-1.1} is used to define some relative operator entropies \cite{FK,FYK}, and recently
called the operator perspective \cite{EH} having a role in operator theory.

In \cite{HMPB,HM} we gave comprehensive expositions on standard and maximal $f$-divergences in
the matrix setting, mostly from the point of view of the \emph{monotonicity inequality} (often
called the \emph{data-processing inequality}) and the \emph{reversibility} of quantum
operations. Our goal in the next stage is to extend expositions in \cite{HMPB,HM} to the von
Neumann algebra setting. We expect that those extensions would be useful in further
developments of quantum information, as well as in some mathematical physics subjects such as
quantum field theory (see \cite{LX} for the appearance of the relative entropy there). In the
previous paper \cite{Hi1} we made systematic study of standard $f$-divergences and standard
R\'enyi divergences in von Neumann algebras. The methodological novelty there is to generalize
Kosaki's variational expression of the relative entropy to general $S_f(\rho\|\sigma)$, from
which many important properties of $S_f$ follow immediately. The aim of the present paper is
to develop maximal $f$-divergences in a similar way in the general von Neumann algebra setting.

In section 2 of this paper we first give the definition of the maximal $f$-divergence
$\widehat S_f(\rho\|\sigma)$ for $\rho,\sigma\in M_*^+$ (the positive part of the predusl $M_*$
of a von Neumann algebra $M$) and for any operator convex function $f$ on $(0,+\infty)$. For
this our idea is to use the representatives $h_\rho,h_\sigma$ in Haagerup's $L^1(M)$ space and
the functional $\tr$ on $L^1(M)$ in place of the trace $\Tr$ in the matrix case. When
$\delta\sigma\le\rho\le\delta^{-1}\sigma$ for some $\delta>0$, there exists a unique
$A\in s(\sigma)Ms(\sigma)$ ($s(\sigma)$ being the support projection of $\sigma$) such that
$h_\rho^{1/2}=Ah_\sigma^{1/2}$, so that we define $\widehat S_f(\rho\|\sigma):=\sigma(f(A^*A))$.
Writing $A^*A=(h_\rho^{1/2}h_\sigma^{-1/2})^*(h_\rho^{1/2}h_\sigma^{-1/2})
=h_\sigma^{-1/2}h_\rho h_\sigma^{-1/2}$ formally, we can write $\widehat S_f(\rho\|\sigma)
=\tr\,h_\sigma^{1/2}f(h_\sigma^{-1/2}h_\rho h_\sigma^{-1/2})h_\sigma^{1/2}$, having a complete
resemblance to \eqref{F-1.1}. We then extend $\widehat S_f(\rho\|\sigma)$ by convergence to
arbitrary $\rho,\sigma\in M_*^+$, and prove the monotonicity inequality under unital positive
normal maps and the joint convexity (Theorem \ref{T-2.9}).

In Sections 3 and 4, we analyze the case where $\rho$ is strongly absolutely continuous with
respect to $\sigma$, and obtain a general integral expression (Theorem \ref{T-4.2})
\begin{align}\label{F-1.3}
\widehat S_f(\rho\|\sigma)=\int_0^1(1-t)f\biggl({t\over1-t}\biggr)
\,d\|E_{\rho/\rho+\sigma}(t)\xi_{\rho+\sigma}\|^2,
\end{align}
where $E_{\rho/\rho+\sigma}(\cdot)$ is the spectral measure of the operator
$T_{\rho/\rho+\sigma}\in\pi_{\rho+\sigma}(M)'_+$ such that
$\rho(x)=\<\xi_{\rho+\sigma},T_{\rho/\rho+\sigma}\pi_{\rho+\sigma}(x)\xi_{\rho+\sigma}\>$,
$x\in M$, for the cyclic representation $\pi_{\rho+\sigma}$ of $M$ associated with $\rho+\sigma$.
In Section 5, the lower semicontinuity in the norm topology and the martingale convergence for
$\widehat S_f(\rho\|\sigma)$ are proved by using the expression in \eqref{F-1.3}.

In Section 6, following Matsumoto's idea in \cite{Ma}, we obtain a variational expression of
the form (Theorem \ref{T-6.3})
\begin{align}\label{F-1.4}
\widehat S_f(\rho\|\sigma)=\min\{S_f(p\|q):(\Psi,p,q)\},
\end{align}
where the minimum is attained over \emph{reverse tests} $(\Psi,p,q)$ for $\rho,\sigma$
consisting of a unital positive normal map $\Psi:M\to L^\infty(X,\mu)$ on a $\sigma$-finite
measure space $(X,\mu)$ and $p,q\in L^1(X,\mu)_+$ with $\Psi_*(p)=\rho$ and $\Psi_*(q)=\sigma$,
and $S_f(p\|q)$ is the classical $f$-divergence of $p,q$. From the variational expression in
\eqref{F-1.4} we have the inequality $S_f(\rho\|\sigma)\le\widehat S_f(\rho\|\sigma)$, in
particular, $D(\rho\|\sigma)\le D_\BS(\rho\|\sigma)$ that was first proved in \cite{HiPe} for
matrices, and moreover the equality $S_f(\rho\|\sigma)=\widehat S_f(\rho\|\sigma)$ is
verified when $\rho,\sigma$ commute. In this way, we present three different expressions of
$\widehat S_f(\rho\|\sigma)$ -- the definition in the beginning, expressions \eqref{F-1.3} and
\eqref{F-1.4}, each of which is useful in deriving some of different properties of
$\widehat S_f(\rho\|\sigma)$.

Finally, the extension of $\widehat S_f(\rho\|\sigma)$ to general positive linear functionals
$\rho,\sigma$ on a unital $C^*$-algebra is discussed in Section 7, and remarks and problems for
further investigation are mentioned in Section 8.

\section{Definition and basic properties}

Let $M$ be a general von Neumann algebra with predual $M_*$, and $M_*^+$ be the positive part
of $M_*$ consisting of normal positive linear functionals on $M$. In the present paper it is
convenient for us to work in the framework of Haagerup's $L^p$-spaces associated with $M$, so
we first recall Haagerup's $L^p$-spaces, see \cite{Te} for details. Given a faithful normal
semifinite weight $\ffi_0$ on $M$, let $N$ denote the crossed product
$M\rtimes_{\sigma^{\ffi_0}}\bR$ of $M$ by the modular automorphism group $\sigma_t^{\ffi_0}$
($t\in\bR$). Let $\theta_s$ ($s\in\bR$) be the dual action of $N$ so that
$\tau\circ\theta_s=e^{-s}\tau$ ($s\in\bR$), where $\tau$ is the canonical trace on $N$. Let
$\widetilde N$ denote the space of $\tau$-measurable operators affiliated with $N$. For
$0<p\le\infty$ \emph{Haagerup's $L^p$-space} $L^p(M)$ \cite{Ha2,Te} is defined by
$$
L^p(M):=\{x\in\widetilde N: \theta_s(x)=e^{-s/p}x,\ s\in\bR\}.
$$
In particular, $L^\infty(M)=M$. Let $L^p(M)_+=L^p(M)\cap\widetilde N_+$, where $\widetilde N_+$
is the positive part of $\widetilde N$. Then $M_*$ is canonically order-isomorphic to
$L^1(M)$ by a linear bijection $\psi\in M_*\mapsto h_\psi\in L^1(M)$, so that the positive
linear functional $\tr$ on $L^1(M)$ is defined by $\tr(h_\psi)=\psi(1)$, $\psi\in M_*$.

For $1\le p<\infty$ the $L^p$-norm $\|x\|_p$ of $x\in L^p(M)$ is given by
$\|x\|_p:=\tr(|x|^p)^{1/p}$. Also $\|\cdot\|_\infty$ denotes the operator norm on $M$. For
$1\le p<\infty$, $L^p(M)$ is a Banach space with the norm $\|\cdot\|_p$ and whose dual Banach
space is $L^q(M)$ where $1/p+1/q=1$ by the duality
$$
(x,y)\in L^p(M)\times L^q(M)\ \longmapsto\ \tr(xy)\ (=\tr(yx)).
$$
In particular, $L^2(M)$ is a Hilbert space with the inner product $\<x,y\>=\tr(x^*y)$
($=\tr(yx^*)$). Then
$$
(M,L^2(M),J=\,^*,L^2(M)_+)
$$
becomes a \emph{standard form} \cite{Ha} of $M$, where $M$ is represented on $L^2(M)$ by the
left multiplication. By the uniqueness of a standard form of $M$ up to unitary equivalence
\cite{Ha}, our discussions in this paper are independent of the choice of a standard form of
$M$. But the standard form $(M,L^2(M),\,^*\,,L^2(M)_+)$ is more convenient since the Haagerup's
$L^p$-space technique is sometimes useful. Each $\sigma\in M_*^+$ is represented as
$$
\sigma(x)=\tr(xh_\sigma)=\<h_\sigma^{1/2},xh_\sigma^{1/2}\>,\qquad x\in M,
$$
with the vector representative $h_\sigma^{1/2}\in L^2(M)_+$. Note that the support projection
$s(\sigma)$ ($\in M$) of $\sigma$ coincides with that of $h_\sigma$. We also note from
\cite[Corollary 2.5, Lemma 2.6]{Ha} that for every projection $e\in M$, the standard form of
the reduced von Neumann algebra $eMe$ is given as
$$
(eMe,eL^2(M)e, J=\,^*, eL^2(M)_+e).
$$

The next lemma is well-known while we give a proof for completeness.

\begin{lemma}\label{L-2.1}
Let $\rho,\sigma\in M_*^+$. Assume that $\rho\le\alpha\sigma$, i.e.,
$h_\rho\le\alpha h_\sigma$ for some $\alpha>0$. Then there exists a unique
$A\in s(\sigma)Ms(\sigma)$ such that $h_\rho^{1/2}=Ah_\sigma^{1/2}$. The $A$ satisfies
$\|A\|\le\alpha^{1/2}$. Moreover, if $\beta\sigma\le\rho\le\alpha\sigma$ for some
$\alpha,\beta>0$, then the above $A$ satisfies $\beta s(\sigma)\le A^*A\le\alpha s(\sigma)$.
\end{lemma}

\begin{proof}
From the assumption, we have $\|h_\rho^{1/2}\xi\|\le\alpha^{1/2}\|h_\sigma^{1/2}\xi\|$ for
all $\xi\in\cD(h_\sigma^{1/2})$, the domain of $h_\sigma^{1/2}$. Since
$h_\rho,h_\sigma\in\widetilde N$ ($\tau$-measurable operators), we have a unique operator
$A\in N$ such that $A(h_\sigma^{1/2}\xi)=h_\rho^{1/2}\xi$ for $\xi\in\cD(h_\sigma^{1/2})$ and
$A(1-s(\sigma))=0$. These imply that $A=s(\sigma)A=As(\sigma)$ and $A^*A\le\alpha s(\sigma)$.
Since $\theta_s(h_\rho^{1/2})=e^{-s/2}h_\rho^{1/2}$ and
$\theta_s(h_\sigma^{1/2})=e^{-s/2}h_\sigma^{1/2}$,
$\theta_s(h_\rho^{1/2})=\theta_s(Ah_\sigma^{1/2})$ means that
$h_\rho^{1/2}=\theta_s(A)h_\sigma^{1/2}$. Hence it follows that $\theta_s(A)=A$ for all
$s\in\bR$, implying that $A\in N^\theta=M$ (where $N^\theta$ is the $\theta$-fixed point
algebra). Therefore, $A\in s(\sigma)Ms(\sigma)$.

Next, assume that $\beta\sigma\le\rho$ in addition to $\rho\le\alpha\sigma$. Then
$s(\rho)=s(\sigma)$, and there is a unique $B\in s(\sigma)Ms(\sigma)$ such that
$h_\sigma^{1/2}=Bh_\rho^{1/2}$. It is easy to see that $AB=BA=s(\sigma)$, hence $B=A^{-1}$ in
$s(\sigma)Ms(\sigma)$. Since $BB^*\le\beta^{-1}s(\sigma)$, we have
$\beta s(\sigma)\le A^*A\le\alpha s(\sigma)$.
\end{proof}

\begin{remark}\label{R-2.2}\rm
We have supplied a rather direct proof of Lemma \ref{L-2.1} for the convenience of the
reader. But a more advanced proof can be given with use of the Connes cocycle Radon-Nikodym
derivative $[D\rho:D\sigma]_t$ \cite{C1}, as in \cite[Theorem VIII.3.17]{Ta2}; see also
\cite[Lemma A.1]{Hi1}. In fact, $\rho\le\alpha\sigma$ for some $\alpha>0$ if and only if
$s(\rho)\le s(\sigma)$ and $[D\rho:D\sigma]_t$ extends to a weakly continuous ($M$-valued)
function $[D\rho:D\sigma]_z$ on the strip $-1/2\le\Im z\le0$ which is analytic in the interior.
In this case, $\|[D\rho:D\sigma]_{-i/2}\|\le\alpha^{1/2}$ and
$h_\rho^{1/2}=[D\rho:D\sigma]_{-i/2}h_\sigma^{1/2}$. So, $A\in s(\sigma)Ms(\sigma)$ given in
Lemma \ref{L-2.1} is $[D\rho:D\sigma]_{-i/2}$, which also shows that the operator $A$ is
determined independently of the choice of the standard form of $M$.
\end{remark}

Throughout the paper we assume that $f$ is an \emph{operator convex function} on $(0,+\infty)$,
i.e., $f$ is a real function on $(0,+\infty)$ such that the operator inequality
$$
f(\lambda A+(1-\lambda)B)\le\lambda f(A)+(1-\lambda)f(B),\qquad 0\le\lambda\le1,
$$
holds for every positive invertible operators $A,B$ on any Hilbert space. We set
$$
f(0^+):=\lim_{x\searrow0}f(x),\qquad f'(+\infty):=\lim_{x\to+\infty}f(x)/x,
$$
which are in $(-\infty,+\infty]$.

For $\rho,\sigma\in M_*^+$ we write $\rho\sim\sigma$ if
$\delta\sigma\le\rho\le\delta^{-1}\sigma$ for some $\delta>0$, and we set
\begin{align*}
(M_*^+\times M_*^+)_\sim&:=\{(\rho,\sigma)\in M_*^+\times M_*^+:
\rho\sim\sigma\}, \\
(M_*^+\times M_*^+)_\le&:=\{(\rho,\sigma)\in M_*^+\times M_*^+:
\rho\le\alpha\sigma\ \mbox{for some $\alpha>0$}\}, \\
(M_*^+\times M_*^+)_\ge&:=\{(\rho,\sigma)\in M_*^+\times M_*^+:
\alpha\rho\ge\sigma\ \mbox{for some $\alpha>0$}\},
\end{align*}
which are all convex sets. We first define the maximal $f$-divergence for
$(\rho,\sigma)\in(M_*^+\times M_*^+)_\sim$ and then extend it to general $\rho,\sigma\in M_*^+$.

\begin{definition}\label{D-2.3}\rm
For each $(\rho,\sigma)\in(M_*^+\times M_*^+)_\sim$ let $A\in s(\sigma)Ms(\sigma)$ be as given
in Lemma \ref{L-2.1}, so that $h_\rho^{1/2}=Ah_\sigma^{1/2}$. Since $A^*A$ is a positive
invertible operator in $s(\sigma)Ms(\sigma)$, we define an self-adjoint operator $f(A^*A)$ in
$s(\sigma)Ms(\sigma)$ via functional calculus. We define the \emph{maximal $f$-divergence}
of $\rho$ with respect to $\sigma$ by
\begin{align}\label{F-2.1}
\widehat S_f(\rho\|\sigma):=\sigma(f(A^*A))\in\bR.
\end{align}
Here the symbol $\widehat S$ is used to distinguish the maximal $f$-divergence from the
standard $f$-divergence $S_f(\rho\|\sigma)$ studied in \cite{Hi1}. Since
$h_\rho=h_\sigma^{1/2}A^*Ah_\sigma^{1/2}$, it is natural to denote $A^*A$ by
$h_\sigma^{-1/2}h_\rho h_\sigma^{-1/2}$ though the expression is rather formal. Below we will
sometimes use this expression. Then \eqref{F-2.1} is rewritten as
\begin{align}\label{F-2.2}
\widehat S_f(\rho\|\sigma)=\tr(h_\sigma f(h_\sigma^{-1/2}h_\rho h_\sigma^{-1/2})),
\end{align}
which is in the same form as the maximal $f$-divergence in the matrix case \cite{HM} if we
consider $\tr$ as the usual trace and $h_\rho,h_\sigma$ as the density matrices.
\end{definition}

\begin{lemma}\label{L-2.4}
Let $M_0$ be another von Neumann algebra and $\Phi:M_0\to M$ be a unital positive map that is
normal (i.e., if $\{x_\alpha\}$ is an increasing net in $M_+$ with $x_\alpha\nearrow x\in M_+$,
then $\Phi(x_\alpha)\nearrow\Phi(x)$). Then for every
$(\rho,\sigma)\in(M_*^+\times M_*^+)_\sim$,
$$
\widehat S_f(\rho\circ\Phi\|\sigma\circ\Phi)\le\widehat S_f(\rho\|\sigma).
$$
\end{lemma}

\begin{proof}
One can define the predual map $\Phi_*:L^1(M)\to L^1(M_0)$ of $\Phi:M_0\to M$ by
$\Phi_*(h_\psi)=h_{\psi\circ\Phi}$, $\psi\in M_*$, since
$\tr(h_{\psi\circ\Phi}y)=\psi(\Phi(y))=(\Phi_*\psi)(y)$, $y\in M_0$. Then $\Phi_*$ is a
$\tr$-preserving positive map since
$$
\tr\,\Phi_*(h_\psi)=\tr\,h_{\psi\circ\Phi}=\psi\circ\Phi(1)=\psi(1)=\tr\,h_\psi.
$$
Let $e:=s(\rho)=s(\sigma)\in M$ and $e_0:=s(\rho\circ\Phi)=s(\sigma\circ\Phi)\in M_0$. For
every $x\in(eMe)_+$, since $h_\sigma^{1/2}xh_\sigma^{1/2}\in L^1(M)_+$ and
$h_\sigma^{1/2}xh_\sigma^{1/2}\le\|x\|h_\sigma$, we have
$\Phi_*(h_\sigma^{1/2}xh_\sigma^{1/2})\le\|x\|\Phi_*(h_\sigma)$. By Lemma \ref{L-2.1}, there is
a unique $b\in M_0$ such that $b(1-e_0)=0$ and
$\Phi_*(h_\sigma^{1/2}xh_\sigma^{1/2})^{1/2}=b\Phi_*(h_\sigma)^{1/2}$. Define
$\Psi(x):=b^*b\in(e_0M_0e_0)_+$. One can easily find that $\Psi(\alpha x)=\alpha\Psi(x)$
and $\Psi(x_1+x_2)=\Psi(x_1)+\Psi(x_2)$ for every $x,x_1,x_2\in M_+$ and $\alpha\ge0$. In fact,
the former is obvious. For the latter, let $b_i\in M_0$ ($i=1,2$) be such that $b_i(1-e_0)=0$
and $\Phi_*(h_\sigma^{1/2}x_ih_\sigma^{1/2})^{1/2}=b_i\Phi_*(h_\sigma)^{1/2}$. Since
$\Phi_*(h_\sigma^{1/2}x_ih_\sigma^{1/2})=\Phi_*(h_\sigma)^{1/2}b_i^*b_i\Phi_*(h_\sigma)^{1/2}$,
one has
$$
\Phi_*(h_\sigma^{1/2}(x_1+x_2)h_\sigma^{1/2})
=\Phi_*(h_\sigma)^{1/2}(b_1^*b_1+b_2^*b_2)\Phi_*(h_\sigma)^{1/2},
$$
which implies that $\Psi(x_1+x_2)=b_1^*b_1+b_2^*b_2=\Psi(x_1)+\Psi(x_2)$. Then $\Psi$ can
extend to a positive linear map $\Psi:eMe\to e_0M_0e_0$. It is clear that $\Psi$ is unital,
i.e., $\Psi(e)=e_0$. By a Jensen inequality due to Choi \cite{Ch}, for
$T:=h_\sigma^{-1/2}h_\rho h_\sigma^{-1/2}$ (i.e., $T=A^*A$ in Definition \ref{D-2.3}) we have
$$
f(\Psi(T))\le \Psi(f(T)).
$$
Since
$$
\Phi_*(h_\sigma)^{1/2}\Psi(T)\Phi_*(h_\sigma)^{1/2}
=\Phi_*(h_\sigma^{1/2}Th_\sigma^{1/2})=\Phi_*(h_\rho),
$$
we have $\Psi(T)=\Phi_*(h_\sigma)^{-1/2}\Phi_*(h_\rho)\Phi_*(h_\sigma)^{-1/2}$ and
$$
\Phi_*(h_\sigma)^{1/2}f(\Psi(T))\Phi_*(h_\sigma)^{1/2}
\le\Phi_*(h_\sigma)^{1/2}\Psi(f(T))\Phi_*(h_\sigma)^{1/2}
=\Phi_*(h_\sigma^{1/2}f(T)h_\sigma^{1/2}).
$$
Therefore,
\begin{align*}
\widehat S_f(\rho\circ\Phi\|\sigma\circ\Phi)
&=\tr\bigl(\Phi_*(h_\sigma)^{1/2}f(\Psi(T))\Phi_*(h_\sigma)^{1/2}\bigr) \\
&\le\tr\bigl(\Phi_*(h_\sigma^{1/2}f(T)h_\sigma^{1/2})\bigr)
=\tr\bigl(h_\sigma^{1/2}f(T)h_\sigma^{1/2}\bigr)
=\widehat S_f(\rho\|\sigma).
\end{align*}
\end{proof}

\begin{lemma}\label{L-2.5}
$\widehat S_f(\rho\|\sigma)$ is jointly convex on $(M_*^+\times M_*^+)_\sim$. Slightly more
strongly, for any $(\rho_i,\sigma_i)\in(M_*^+\times M_*^+)_\sim$ and $\lambda_i\ge0$
($1\le i\le n$) we have
$$
\widehat S_f\Biggl(\sum_{i=1}^n\lambda_i\rho_i\Bigg\|\sum_{i=1}^n\lambda_i\sigma_i\Biggr)
\le\sum_{i=1}^n\lambda_i\widehat S_f(\rho_i\|\sigma_i).
$$
\end{lemma}

\begin{proof}
Let $\cM:=\oplus_{i=1}^nM$ and $\Phi:M\to\cM$ be a unital *-homomorphism (hence, completely
positive) given as $\Phi(x):=x\oplus\cdots\oplus x$, $x\in M$. Note that the standard form
of $\cM$ is given as the direct sum $\oplus_{i=1}^n(M,L^2(M),\,^*,L^2(M)_+)$. For given
$(\rho_i,\sigma_i)\in(M_*^+\times M_*^+)_\sim$ and $\lambda_i$ let
$\rho:=\oplus_{i=1}^n\lambda_i\rho_i$ and $\sigma:=\oplus_{i=1}^n\lambda_i\sigma_i$ in
$\cM_*^+$, so $(\rho,\sigma)\in(\cM_*^+\times\cM_*^+)_\sim$. Since
$\rho\circ\Phi=\sum_{i=1}^n\lambda_i\rho_i$ and
$\sigma\circ\Phi=\sum_{i=1}^n\lambda_i\sigma_i$, Lemma \ref{L-2.4} yields
$$
\widehat S_f\Biggl(\sum_i\lambda_i\rho_i\Bigg\|\sum_i\lambda_i\sigma_i\Biggr)
\le\widehat S_f(\rho\|\sigma).
$$
Since $h_\rho=\oplus_{i=1}^n\lambda_ih_{\rho_i}$ and
$h_\sigma=\oplus_{i=1}^n\lambda_ih_{\sigma_i}$, it is immediate to see that
$$
\widehat S_f(\rho\|\sigma)=\sum_i\lambda_i\widehat S_f(\rho_i\|\sigma_i),
$$
implying the asserted inequality.
\end{proof}

To extend the maximal $f$-divergence $\widehat S_f(\rho\|\sigma)$ to arbitrary
$\rho,\sigma\in M_*^+$, we give the following:

\begin{lemma}\label{L-2.6}
Let $\rho,\sigma\in M_*^+$. For every $\eta\in M_*^+$ with $\eta\sim\rho+\sigma$, the limit
\begin{align}\label{F-2.3}
\lim_{\eps\searrow0}\widehat S_f(\rho+\eps\eta\|\sigma+\eps\eta)\in(-\infty,+\infty]
\end{align}
exists, and moreover the limit is independent of the choice of $\eta$ as above.
\end{lemma}

\begin{proof}
Let $\eta$ be given as stated. Since $\rho+\eps\eta\sim\sigma+\eps\eta$,
$\widehat S_f(\rho+\eps\eta\|\sigma+\eps\eta)$ is defined for each $\eps>0$ by Definition
\ref{D-2.3}, and $0<\eps\mapsto\widehat S_f(\rho+\eps\eta\|\sigma+\eps\eta)$ is convex by
Lemma \ref{L-2.5}. Hence the limit in \eqref{F-2.3} exists in $(-\infty,+\infty]$.

To prove the independence of the choice of $\eta$, let $\eta_1,\eta_2\in M_*^+$ be such that
$\eta_i\sim\rho+\sigma$ ($i=1,2$). Choose a $\delta>0$ such that
$\delta\eta_1\le\eta_2\le\delta^{-1}\eta_1$. By Lemma \ref{L-2.5} we have
\begin{align*}
\widehat S_f(\rho+\eps\eta_1\|\sigma+\eps\eta_1)
&=\widehat S_f(\rho+\eps\delta\eta_2+\eps(\eta_1-\delta\eta_2)\|
\sigma+\eps\delta\eta_2+\eps(\eta_1-\delta\eta_2)) \\
&\le\widehat S_f(\rho+\eps\delta\eta_2\|\sigma+\eps\delta\eta_2)
+\widehat S_f(\eps(\eta_1-\delta\eta_2)\|\eps(\eta_1-\delta\eta_2)) \\
&=\widehat S_f(\rho+\eps\delta\eta_2\|\sigma+\eps\delta\eta_2)
+\eps(\eta_1-\delta\eta_2)(1)f(1).
\end{align*}
Therefore,
$$
\lim_{\eps\searrow0}\widehat S_f(\rho+\eps\eta_1\|\sigma+\eps\eta_1)
\le\lim_{\eps\searrow0}\widehat S_f(\rho+\eps\delta\eta_2\|\sigma+\eps\delta\eta_2)
=\lim_{\eps\searrow0}\widehat S_f(\rho+\eps\eta_2\|\sigma+\eps\eta_2).
$$
The converse inequality is similar.
\end{proof}

\begin{lemma}\label{L-2.7}
If $(\rho,\sigma)\in(M_*^+\times M_*^+)_\sim$, then
$$
\widehat S_f(\rho\|\sigma)=\lim_{\eps\searrow0}\widehat S_f(\rho+\eps\sigma\|(1+\eps)\sigma).
$$
\end{lemma}

\begin{proof}
We find that
\begin{align*}
\widehat S_f(\rho+\eps\sigma\|(1+\eps)\sigma)
&=(1+\eps)\sigma\biggl(f\biggl({h_\sigma^{-1/2}(h_\rho+\eps h_\sigma)h_\sigma^{-1/2}
\over1+\eps}\biggr)\biggr) \\
&=(1+\eps)\sigma\biggl(f\biggl({h_\sigma^{-1/2}h_\rho h_\sigma^{-1/2}+\eps s(\sigma)
\over 1+\eps}\biggr)\biggr).
\end{align*}
Let $h_\sigma^{-1/2}h_\rho h_\sigma^{-1/2}=\int_\delta^{\delta^{-1}}t\,dE(t)$ be the spectral
decomposition, where $0<\delta<1$ and $\int_\delta^{\delta^{-1}}dE(t)=s(\sigma)$. Then we see
that
$$
f\biggl({h_\sigma^{-1/2}h_\rho h_\sigma^{-1/2}+\eps s(\sigma)\over 1+\eps}\biggr)
=\int_\delta^{\delta^{-1}}f\biggl({t+\eps\over1+\eps}\biggr)\,dE(t)
$$
converges to $\int_\delta^{\delta^{-1}}f(t)\,dE(t)=f(h_\sigma^{-1/2}h_\rho h_\sigma^{-1/2})$
in the operator norm, so that
$$
\lim_{\eps\searrow0}\widehat S_f(\rho+\eps\sigma\|(1+\eps)\sigma)
=\sigma(f(h_\sigma^{-1/2}h_\rho h_\sigma^{-1/2})=\widehat S_f(\rho\|\sigma).
$$
\end{proof}

\begin{definition}\label{D-2.8}\rm
For every $\rho,\sigma\in M_*^+$ define the \emph{maximal $f$-divergence}
$\widehat S_f(\rho\|\sigma)$ by
\begin{align}\label{F-2.4}
\widehat S_f(\rho\|\sigma)
:=\lim_{\eps\searrow0}\widehat S_f(\rho+\eps\eta\|\sigma+\eps\eta)\in(-\infty,+\infty]
\end{align}
for any $\eta\in M_*^+$ with $\eta\sim\rho+\sigma$, where
$\widehat S_f(\rho+\eps\eta\|\sigma+\eps\eta)$ is defined in Definition \ref{D-2.3}. By Lemmas
\ref{L-2.6} and \ref{L-2.7} the definition is well defined independently of the choice of
$\eta$ and extend Definition \ref{D-2.3} for the case $\rho\sim\sigma$.
\end{definition}

\begin{thm}\label{T-2.9}
The monotonicity property of Lemma \ref{L-2.4} and the joint convexity property of Lemma
\ref{L-2.5} hold true for $\widehat S_f(\rho\|\sigma)$ for general $\rho,\sigma\in M_*^+$.
\end{thm}

\begin{proof}
Let $\Phi:M_0\to M$ be as in Lemma \ref{L-2.4}. For every $\rho,\sigma\in M_*^+$, by Lemma
\ref{L-2.4} we have for every $\eps>0$,
$$
\widehat S_f((\rho+\eps(\rho+\sigma))\circ\Phi\|(\sigma+\eps(\rho+\sigma))\circ\Phi)
\le\widehat S_f(\rho+\eps(\rho+\sigma)\|\sigma+\eps(\rho+\sigma)).
$$
Thanks to Definition \ref{D-2.8}, letting $\eps\searrow0$ gives
$\widehat S_f(\rho\circ\Phi\|\sigma\circ\Phi)\le\widehat S_f(\rho\|\sigma)$.

For any $\rho_i,\sigma_i\in M_*^+$ and $\lambda_i\ge0$ ($1\le i\le n$), by Lemma \ref{L-2.5}
we have for every $\eps>0$,
\begin{align*}
&\widehat S_f\Biggl(\sum_i\lambda_i\rho_i+\eps\Biggl(
\sum_i\lambda_i\rho_i+\sum_i\lambda_i\sigma_i\Biggr)\Bigg\|
\sum_i\lambda_i\sigma_i+\eps\Biggl(
\sum_i\lambda_i\rho_i+\sum_i\lambda_i\sigma_i\Biggr)\Biggr) \\
&\quad=\widehat S_f\Biggl(\sum_i\lambda_i(\rho_i+\eps(\rho_i+\sigma_i))\Bigg\|
\sum_i\lambda_i(\sigma_i+\eps(\rho_i+\sigma_i))\Biggr) \\
&\quad\le\sum_i\lambda_i\widehat S_f(\rho_i+\eps(\rho_i+\sigma_i)\|
\sigma_i+\eps(\rho_i+\sigma_i)).
\end{align*}
Letting $\eps\searrow0$ gives $\widehat S_f(\sum_i\lambda_i\rho_i\|\sum_i\lambda_i\sigma_i)
\le\sum_i\lambda_i\widehat S_f(\rho_i\|\sigma_i)$.
\end{proof}

Another significant property of $\widehat S_f(\rho\|\sigma)$ is the joint lower semicontinuity,
which we will prove later in Section 5 after developing a general integral formula in Section 4.

The \emph{transpose} $\widetilde f$ of $f$ is defined by
$$
\widetilde f(t):=tf(t^{-1}),\qquad t\in(0,+\infty),
$$
which is again an operator convex function on $(0,+\infty)$, see \cite[Proposition A.1]{HM}.
It is immediate to see that $\widetilde f(0^+)=f'(+\infty)$ and $\widetilde f'(+\infty)=f(0^+)$.
The next proposition shows the symmetry of $\widehat S_f(\rho\|\sigma)$ between two variables
under exchanging $f$ and $\widehat f$.

\begin{prop}\label{P-2.10}
For every $\rho,\sigma\in M_*^+$,
$$
\widehat S_{\widetilde f}(\rho\|\sigma)=\widehat S_f(\sigma\|\rho).
$$
\end{prop}

\begin{proof}
Assume first that $(\rho,\sigma)\in(M_*^+\times M_*^+)_\sim$. Let $A,B\in s(\sigma)Ms(\sigma)$
be as given in the proof of Lemma \ref{L-2.1}. Then we write
$$
\widehat S_{\widetilde f}(\rho\|\sigma)
=\tr\,h_\sigma^{1/2}\widetilde f(A^*A)h_\sigma^{1/2}
$$
and
$$
\widehat S_f(\sigma\|\rho)=\tr\,h_\rho^{1/2}f(B^*B)h_\rho^{1/2}
=\tr\,h_\sigma^{1/2}A^*f(B^*B)Ah_\sigma^{1/2}.
$$
Hence it suffices to show that
\begin{align}\label{F-2.5}
\widetilde f(A^*A)=A^*f(B^*B)A
\end{align}
for every continuous function $f$ on $(0,+\infty)$. By approximation we may show \eqref{F-2.5}
when $f(t)=t^m$ for any non-negative integer $m$. Since $A=B^{-1}$, we have
\begin{align*}
A^*f(B^*B)A&=B^{*-1}(B^*B)^mB^{-1}=(BB^*)^{m-1} \\
&=((A^*A)^{-1})^{m-1}=A^*A((A^*A)^{-1})^m=\widetilde f(A^*A),
\end{align*}
so that \eqref{F-2.5} is shown. Now, the asserted equality for general $\rho,\sigma\in M_*^+$
immediately follows from the above case and Definition \ref{D-2.8}.
\end{proof}

\begin{prop}\label{P-2.11}
Let $\rho_i,\sigma_i\in M_*^+$ ($i=1,2$). If
$s(\rho_1)\vee s(\sigma_1)\perp s(\rho_2)\vee s(\sigma_2)$, then
$$
\widehat S_f(\rho_1+\rho_2\|\sigma_1+\sigma_2)
=\widehat S_f(\rho_1\|\sigma_1)+\widehat S_f(\rho_2\|\sigma_2).
$$
\end{prop}

\begin{proof}
In view of Definition \ref{D-2.8} one may show the identity in the case where
$(\rho_i,\sigma_i)\in(M_*^+\times M_*^+)_\sim$ ($i=1,2$) with $s(\sigma_1)\perp s(\sigma_2)$.
For $i=1,2$ choose an $A_i\in s(\sigma_i)Ms(\sigma_i)$ such that
$h_{\rho_i}^{1/2}=A_ih_{\sigma_i}^{1/2}$. Note that
\begin{align*}
(A_1+A_2)h_{\sigma_1+\sigma_2}^{1/2}
&=(A_1+A_2)(h_{\sigma_1}^{1/2}+h_{\sigma_2}^{1/2})
=A_1h_{\sigma_1}^{1/2}+A_2h_{\sigma_2}^{1/2} \\
&=h_{\rho_1}^{1/2}+h_{\rho_2}^{1/2}=h_{\rho_1+\rho_2}^{1/2}
\end{align*}
and $f((A_1+A_2)^*(A_1+A_2))=f(A_1^*A_1)+f(A_2^*A_2)$ as operators in
$s(\sigma_1+\sigma_2)Ms(\sigma_1+\sigma_2)$. Hence the asserted equality follows.
\end{proof}

\begin{example}\label{E-2.12}\rm
When $M$ is semifinite with a faithful normal semifinite trace $\tau_0$, we have the
conventional non-commutative $L^p$-space $L^p(M,\tau_0)$ for $1\le p<\infty$, the space of
$\tau_0$-measurable operators $x$ affiliated with $M$ such that
$\|x\|_p^p=\tau_0(|x|^p)<+\infty$, see \cite{Ne}. The explicit relation between $L^p(M,\tau_0)$
and Haagerup's $L^p(M)$ is found in \cite[pp.\ 62--63]{Te}. In the semifinite case, $M$ is
standardly represented on the Hilbert space $L^2(M,\tau_0)$ by the left multiplication, and one
can define $\widehat S_f(\rho\|\sigma)$ for $\rho,\sigma$ with use of Radon-Nikodym derivatives
$h_\rho:=d\rho/d\tau_0$, $h_\sigma:=d\sigma/d\tau_0$ in $L^1(M,\tau_0)_+$ (so
$\rho(x)=\tau_0(x h_\rho)$ for $x\in M$) in place of Haagerup's $h_\rho,h_\sigma$.

In particular, assume that $M$ is the algebra $\cB(\cH)$ of all linear operators on a
finite-dimensional Hilbert space $\cH$, or $M=\bM_d$, the matrix algebra of size $d:=\dim\cH$.
Let $\rho,\sigma\in\bM_d^+$, which define positive linear functionals $\rho(X):=\Tr\rho X$,
$\sigma(X):=\Tr\sigma X$ for $X\in\bM_d$ with the same notations as $\rho,\sigma$, where $\Tr$
is the usual matrix trace. Then $\widehat S_f(\rho\|\sigma)$ coincides with that defined in
\cite[Definition 3.21]{HM}. In fact, when $\rho,\sigma$ are invertible, Definition \ref{D-2.3}
becomes $\widehat S_f(\rho\|\sigma)=\Tr\sigma f(\sigma^{-1/2}\rho\sigma^{-1/2})$. For general
$\rho,\sigma\in\bM_d^+$ let $e$ be the support projection of $\rho+\sigma$. By Proposition
\ref{P-2.11}, $\widehat S_f(\rho\|\sigma)$ in \cite{HM} is defined as
\begin{align*}
\lim_{\eps\searrow0}\widehat S_f(\rho+\eps I\|\sigma+\eps I)
&=\lim_{\eps\searrow0}\bigl\{\widehat S_f(\rho+\eps e\|\sigma+\eps e)
+\widehat S_f(\eps(I-e)\|\eps(I-e))\bigr\} \\
&=\lim_{\eps\searrow0}\bigl\{\widehat S_f(\rho+\eps e\|\sigma+\eps e)
+\eps f(1)\sigma(I-e)\bigr\} \\
&=\lim_{\eps\searrow0}\widehat S_f(\rho+\eps e\|\sigma+\eps e),
\end{align*}
which is Definition \ref{D-2.8}.
\end{example}

\begin{example}\label{E-2.13}\rm
Let $M$ be an abelian von Neumann algebra such that $M\cong L^\infty(X,\mu)$ on a
$\sigma$-finite measure space $(X,\cX,\mu)$. The standard form of $M\cong L^\infty(X,\mu)$ is
$(L^\infty(X,\mu),L^2(X,\mu),\allowbreak\xi\mapsto\overline\xi,L^2(X,\mu)_+)$, where
$\phi\in L^\infty(X,\mu)$ is represented on $L^2(X,\mu)$ as the multiplication operator
$\xi\mapsto \phi\xi$, $\xi\in L^2(X,\mu)$. Let $\rho,\sigma\in M_*^+$, which are
identified with functions in $L^1(X,\mu)_+$ (denoted here by the same $\rho,\sigma$ instead
of $h_\rho,h_\sigma$) so that $\rho(\phi)=\int_X\phi\rho\,d\mu$,
$\sigma(\phi)=\int_X\phi\sigma\,d\mu$ for $\phi\in L^\infty(X,\mu)$. With $\eta=\rho+\sigma$
note that
$$
\widehat S_f(\rho+\eps\eta\|\sigma+\eps\eta)
=\int_X(\sigma(x)+\eps\eta(x))f\biggl({\rho(x)+\eps\eta(x)
\over\sigma(x)+\eps\eta(x)}\biggr)\,d\mu(x)
=S_f(\rho+\eps\eta\|\sigma+\eps\eta)
$$
by \cite[Example 2.5]{Hi1}. By Definition \ref{D-2.8} and \cite[Corollary 4.4\,(3)]{Hi1},
take the limit of the above as $\eps\searrow0$ to see that $\widehat S_f(\rho\|\sigma)$
coincides with the classical
$f$-divergence $S_f(\rho\|\sigma)=\int_X\sigma f(\rho/\sigma)\,d\mu$.
\end{example}

\begin{example}\label{E-2.14}\rm
Consider a linear function $f(t)=a+bt$ with $a,b\in\bR$. For every
$(\rho,\sigma)\in(M_*^+\times M_*^+)_\sim$ let $A\in eMe$ be as in Lemma \ref{L-2.1} where
$e=s(\rho)=s(\sigma)$, so that $h_\rho^{1/2}=Ah_\sigma^{1/2}$. Then
\begin{align*}
\widehat S_{a+bt}(\rho\|\sigma)&=\sigma(ae+bA^*A)
=a\sigma(e)+b\,\tr(h_\sigma^{1/2}A^*Ah_\sigma^{1/2}) \\
&=a\sigma(1)+b\,\tr\,h_\rho=a\sigma(1)+b\rho(1).
\end{align*}
This holds for all $\rho,\sigma\in M_*^+$ by Definition \ref{D-2.8}. Hence together with
\cite[(2,7)]{Hi1},
\begin{align}\label{F-2.6}
\widehat S_{a+bt}(\rho\|\sigma)=S_{a+bt}(\rho\|\sigma)=a\sigma(1)+b\rho(1),\qquad
\rho,\sigma\in M_*^+.
\end{align}

Next, consider $f(t)=t^2$. Let $(\rho,\sigma)\in(M_*^+\times M_*^+)_\sim$ and $A$ be as above.
Then
$$
\widehat S_{t^2}(\rho\|\sigma)=\tr(h_\sigma^{1/2}(A^*A)^2h_\sigma^{1/2})
=\tr(h_\rho^{1/2}AA^*h_\rho^{1/2})=\|h_\rho^{1/2}A\|_2^2.
$$
On the other hand, since $h_\rho=(h_\rho^{1/2}A)h_\sigma^{1/2}$, by
\cite[Lemma 5.2 and Proposition A.4\,(2)]{Hi1} we note that
$S_{t^2}(\rho\|\sigma)=\|h_\rho^{1/2}A\|_2^2$. Hence
$\widehat S_{t^2}(\rho\|\sigma)=S_{t^2}(\rho\|\sigma)$ for all $\rho,\sigma\in M_*^+$ by
Definition \ref{D-2.8} and \cite[(4.6)]{Hi1}. Thus, $\widehat S_f=S_f$ if $f$ is a quadratic
polynomial.
\end{example}

In the rest of this section we will present more formulas of $\widehat S_f(\rho\|\sigma)$ in
some special situations.

\begin{prop}\label{P-2.15}
For every $(\rho,\sigma)\in(M_*^+\times M_*^+)_\le$,
$$
\widehat S_f(\rho\|\sigma)=\lim_{\eps\searrow0}\widehat S_f(\rho+\eps\sigma\|\sigma).
$$
For every $(\rho,\sigma)\in(M_*^+\times M_*^+)_\ge$,
$$
\widehat S_f(\rho\|\sigma)=\lim_{\eps\searrow0}\widehat S_f(\rho\|\sigma+\eps\rho).
$$
\end{prop}

\begin{proof}
To show the first assertion, let $(\rho,\sigma)\in(M_*^+\times M_*^+)_\le$; then
$\rho\le\alpha\sigma$ for some $\alpha>0$. Let
$h_\sigma^{-1/2}h_\rho h_\sigma^{-1/2}=\int_0^\alpha t\,dE(t)$ be the spectral decomposition
with $\int_0^\alpha dE(t)=s(\sigma)$. Then, as in the proof of Lemma \ref{L-2.7}, one has
$$
\widehat S_f(\rho+\eps\sigma\|(1+\eps)\sigma)
=(1+\eps)\int_0^\alpha f\biggl({t+\eps\over1+\eps}\biggr)\,d\sigma(E(t)).
$$
Hence, by Definition \ref{D-2.8},
$$
\widehat S_f(\rho\|\sigma)
=\lim_{\eps\searrow0}\widehat S_f(\rho+\eps\sigma\|(1+\eps)\sigma)
=\lim_{\eps\searrow0}\int_0^\alpha f\biggl({t+\eps\over1+\eps}\biggr)\,d\sigma(E(t)).
$$
Since $\widehat S_f(\rho+\eps\sigma\|\sigma)=\int_0^\alpha f(t+\eps)\,d\sigma(E(t))$, it
suffices to show that
\begin{align}\label{F-2.7}
\lim_{\eps\searrow0}\int_0^\alpha f\biggl({t+\eps\over1+\eps}\biggr)\,d\sigma(E(t))
=\lim_{\eps\searrow0}\int_0^\alpha f(t+\eps)\,d\sigma(E(t)).
\end{align}
Consider $f$ as a function on $[0,+\infty)$ by letting $f(0)=f(0^+)\in(-\infty,+\infty]$.
When $f(0^+)<+\infty$, both sides of \eqref{F-2.7} are equal to
$\int_0^\alpha f(t)\,d\sigma(E(t))$ by the bounded convergence theorem. When $f(0^+)=+\infty$,
choose a $\delta>0$ with $\delta\le\min\{\alpha,1\}$ such that $f(t)$ is decreasing on
$(0,\delta)$. Since $f(t+\eps)$ and $f\bigl({t+\eps\over1+\eps}\bigr)$ are increasing to
$f(t)$ as $\delta/2>\eps\searrow0$ for any $t\in[0,\delta/2]$, the monotone convergence
theorem gives
$$
\lim_{\eps\searrow0}\int_0^{\delta/2}f\biggl({t+\eps\over1+\eps}\biggr)\,d\sigma(E(t))
=\lim_{\eps\searrow0}\int_0^{\delta/2}f(t+\eps)\,d\sigma(E(t))
=\int_0^{\delta/2}f(t)\,d\sigma(E(t)).
$$
On the other hand, the bounded convergence theorem gives
$$
\lim_{\eps\searrow0}\int_{\delta/2}^\alpha f\biggl({t+\eps\over1+\eps}\biggr)\,d\sigma(E(t))
=\lim_{\eps\searrow0}\int_{\delta/2}^\alpha f(t+\eps)\,d\sigma(E(t))
=\int_{\delta/2}^\alpha f(t)\,d\sigma(E(t)).
$$
Hence \eqref{F-2.7} follows.

The second assertion is immediate from the first and Proposition \ref{P-2.10}.
\end{proof}

\begin{prop}\label{P-2.16}
If $f(0^+)<+\infty$, then expression \eqref{F-2.1} (or \eqref{F-2.2}) holds for every
$(\rho,\sigma)\in(M_*^+\times M_*^+)_\le$, where $f(0)=f(0^+)$.
\end{prop}

\begin{proof}
Let $(\rho,\sigma)\in(M_*^+\times M_*^+)_\le$. When $f(0^+)<+\infty$, the proof of Proposition
\ref{P-2.15} gives
$$
\widehat S_f(\rho\|\sigma)=\int_0^\alpha f(t)\,d\sigma(E(t))
=\sigma(f(h_\sigma^{-1/2}h_\rho h_\sigma^{-1/2}),
$$
which shows the assertion.
\end{proof}

\begin{prop}\label{P-2.17}
Let $\rho,\sigma\in M_*^+$. If $s(\rho)\not\le s(\sigma)$ and $f'(+\infty)=+\infty$, then
$\widehat S_f(\rho\|\sigma)=+\infty$. If $s(\sigma)\not\le s(\rho)$ and $f(0+)=+\infty$, then
$\widehat S_f(\rho\|\sigma)=+\infty$.
\end{prop}

\begin{proof}
For the first assertion, let $e:=s(\sigma)$ and define a unital positive map $\Phi:\bC^2\to M$
given by $\Phi(a,b):=ae+b(1-e)$ for $(a,b)\in\bC^2$. From the monotonicity property in Theorem
\ref{T-2.9} and Example \ref{E-2.13} we have
$$
\widehat S_f(\rho\|\sigma)\ge\widehat S_f(\rho\circ\Phi\|\sigma\circ\Phi)
=S_f((\rho(e),\rho(1-e))\|(1,0)).
$$
Since $\rho(1-e)>0$, from \cite[Example 2.5]{Hi1} the above right-hans side is
$$
f(\rho(e))+f'(+\infty)\rho(1-e)=+\infty,
$$
where $f(\rho(e))$ means $f(0^+)$ if $\rho(e)=0$. The second assertion follows from the first
and Proposition \ref{P-2.10}.
\end{proof}

\section{Strongly absolutely continuous case}

Let $\rho,\sigma\in M_*^+$. We say that $\rho$ is {\it strongly absolutely continuous} with
respect to $\sigma$ if $\lim_n\rho(x_n^*x_n)=0$ for any sequence $\{x_n\}$ in $M$ such that
$$
\lim_n\sigma(x_n^*x_n)=\lim_{n,m}\rho((x_n-x_m)^*(x_n-x_m))=0.
$$
In this case we write $\rho\ll\sigma$ strongly. Obviously, this implies the simple absolute
continuity, i.e., $\sigma(x^*x)=0$ implies $\rho(x^*x)=0$, equivalently $s(\rho)\le s(\sigma)$.
In terms of $h_\rho$ and $h_\sigma$, we can rewrite the definition of $\rho\ll\sigma$ strongly
as follows: for $\{x_n\}$ in $M$,
$$
\lim_n\|x_nh_\sigma^{1/2}\|_2=\lim_{n,m}\|x_nh_\rho^{1/2}-x_mh_\rho^{1/2}\|_2=0
\ \,\implies\ \,\lim_n\|x_nh_\rho^{1/2}\|_2=0.
$$
This means that the operator
\begin{align}\label{F-3.1}
R=R_{\rho/\sigma}:\ xh_\sigma^{1/2}+\zeta\ \,\longmapsto\ \,xh_\rho^{1/2},
\qquad x\in M,\ \mbox{$\zeta\in L^2(M)$ with $e'\zeta=0$},
\end{align}
is closable, where $e'$ is the projection onto $\overline{Mh_\sigma^{1/2}}$. Note that
$e'=Js(\sigma)J\in M'$ is the $M'$-support of $\sigma$ while $s(\sigma)$ is the projection onto
$\overline{M'h_\sigma^{1/2}}$, the $M$-support of $\sigma$.

We here give the next lemma for completeness; see \cite{Nau,Gudd} for similar
characterizations in a bit more general settings.

\begin{lemma}\label{L-3.1}
The following conditions for $\rho,\sigma\in M_*^+$ are equivalent:
\begin{itemize}
\item[(i)] $\rho\ll\sigma$ strongly;
\item[(ii)] there exists a (unique) positive self-adjoint operator $T=T_{\rho/\sigma}$ on
$L^2(M)$ affiliated with $M'$ such that $Mh_\sigma^{1/2}$ is a core of $T^{1/2}$ and
\begin{align}\label{F-3.2}
\rho(x)=\<T^{1/2}h_\sigma^{1/2},T^{1/2}xh_\sigma^{1/2}\>,\qquad x\in M.
\end{align}
\end{itemize}
\end{lemma}

\begin{proof}
(i)$\implies$(ii).\enspace
Condition (i) means that $R$ given in \eqref{F-3.1} is closable, so let $\overline R$ be its
closure and $\overline R=VT^{1/2}$ be the polar decomposition, where $T:=R^*\overline R$. For
every $x,u\in M$ with $u$ unitary, since
$$
uRu^*xh_\sigma^{1/2}=uu^*xh_\rho^{1/2}=xh_\rho^{1/2}=Rxh_\sigma^{1/2},
$$
we have $uRu^*=R$, so $u\overline Ru^*=\overline R$ and hence $uTu^*=T$. Therefore, $T$ is
affiliated with $M'$ and $V\in M'$. By definition, $Mh_\sigma^{1/2}$ is a core of
$\overline R$. Since any core of $\overline R$ is a core of $T^{1/2}$, $Mh_\sigma^{1/2}$ is a
core of $T^{1/2}$. Moreover, for every $x\in M$,
\begin{align*}
\rho(x)&=\<h_\rho^{1/2},xh_\rho^{1/2}\>
=\<VT^{1/2}h_\sigma^{1/2},xVT^{1/2}h_\sigma^{1/2}\>
=\<T^{1/2}h_\sigma^{1/2},xV^*VT^{1/2}h_\sigma^{1/2}\> \\
&=\<T^{1/2}h_\sigma^{1/2},xT^{1/2}h_\sigma^{1/2}\>
=\<T^{1/2}h_\sigma^{1/2},T^{1/2}xh_\sigma^{1/2}\>,
\end{align*}
where the last equality follows from $xT^{1/2}h_\sigma^{1/2}=T^{1/2}xh_\sigma^{1/2}$ since
$T^{1/2}$ is affiliated with $M'$. The uniqueness of $T$ follows since a positive self-adjoint
operator $T$ is determined by the quadratic form $\|T^{1/2}xh_\sigma^{1/2}\|^2=\rho(x^*x)$,
$x\in M$ (see, e.g., \cite[A.7]{St}).

(ii)$\implies$(i).\enspace
Since \eqref{F-3.2} means that $\|xh_\rho^{1/2}\|=\|T^{1/2}xh_\sigma^{1/2}\|$ for all
$x\in M$, the implication immediately follows since $T^{1/2}$ is a closed operator.
\end{proof}

In the rest of this section, we assume that
$$
\mbox{$f$ is an operator convex function on $(0,+\infty)$ such that $f(0^+)<+\infty$,}
$$
so $f$ extends by continuity to an operator convex function on $[0,+\infty)$. We give the next
definition, following the spirit of Belavkin and Staszewski's relative entropy in \cite{BS}
(also Example \ref{E-3.5}).

\begin{definition}\label{D-3.2}\rm
Let $\rho,\sigma\in M_*^+$ be such that $\rho\ll\sigma$ strongly, and $T_{\rho/\sigma}$ be as
given in Lemma \ref{L-3.1}. We then define
\begin{align}\label{F-3.3}
\widehat S'_f(\rho\|\sigma):=\<h_\sigma^{1/2},f(T_{\rho/\sigma})h_\sigma^{1/2}\>
=\int_0^\infty f(t)\,d\|E_{\rho/\sigma}(t)h_\sigma^{1/2}\|^2,
\end{align}
where $T_{\rho/\sigma}=\int_0^\infty t\,dE_{\rho/\sigma}(t)$ is the spectral decomposition of
$T_{\rho/\sigma}$. The inner product expression in \eqref{F-3.3} should be understood, to be
precise, in the sense of a lower-bounded form (see \cite{RS}), which equals the integral
expression in \eqref{F-3.3}. Since $f(t)\ge at+b$ for some $a,b\in\bR$ and
$$
\int_0^\infty t\,d\|E_{\rho/\sigma}(t)h_\sigma^{1/2}\|^2
=\|T_{\rho/\sigma}^{1/2}h_\sigma^{1/2}\|^2<+\infty,
$$
note that $\widehat S'_f(\rho\|\sigma)$ is well defined with value in $(-\infty,+\infty]$.
\end{definition}

\begin{lemma}\label{L-3.3}
For every $(\rho,\sigma)\in(M_*^+\times M_*^+)_\le$,
$$
\widehat S_f(\rho\|\sigma)=\widehat S'_f(\rho\|\sigma).
$$
\end{lemma}

\begin{proof}
Let $A\in eMe$ be as given in Lemma \ref{L-2.1}, where $e:=s(\sigma)$. Take the polar
decomposition $A=v|A|$, and let $v':=JvJ\in M'$ and $B':=J|A|J\in e'M'e'$ with $e':=JeJ$, the
projection onto $\overline{Mh_\sigma^{1/2}}$ (the $M'$-support of $\sigma$). Note that
$vv^*=s(\rho)$ and so $v'v^{\prime*}=Js(\rho)J$. For every $x\in M$,
$$
B'xh_\sigma^{1/2}=J|A|Jxh_\sigma^{1/2}=Jv^*(Ah_\sigma^{1/2}x^*)
=Jv^*(h_\rho^{1/2}x^*)=v^{\prime*}J(h_\rho^{1/2}x^*)=v^{\prime*}xh_\rho^{1/2}
$$
so that $v'B'=\overline R_{\rho/\sigma}$. In particular,
$B'h_\sigma^{1/2}=v^{\prime*}h_\rho^{1/2}$ and
$$
v'v^{\prime*}h_\rho^{1/2}=Jvv^*h_\rho^{1/2}=Jh_\rho^{1/2}=h_\rho^{1/2}.
$$
Hence we have for every $x\in M$,
\begin{align*}
\<B'h_\sigma^{1/2},B'xh_\sigma^{1/2}\>&=\<B'h_\sigma^{1/2},xB'h_\sigma^{1/2}\>
=\<v^{\prime*}h_\rho^{1/2},xv^{\prime*}h_\rho^{1/2}\> \\
&=\<h_\rho^{1/2},xv'v^{\prime*}h_\rho^{1/2}\>=\<h_\rho^{1/2},xh_\rho^{1/2}\>=\rho(x).
\end{align*}
From the uniqueness assertion in Lemma \ref{L-3.1} this implies that
$B'=T_{\rho/\sigma}^{1/2}$. Therefore,
\begin{align*}
\widehat S'_f(\rho\|\sigma)&=\<h_\sigma^{1/2},f(B^{\prime2})h_\sigma^{1/2}\>
=\<h_\sigma^{1/2},Jf(|A|^2)Jh_\sigma^{1/2}\> \nonumber\\
&=\<f(|A|^2)h_\sigma^{1/2},h_\sigma^{1/2}\>=\sigma(f(A^*A))
=\widehat S_f(\rho\|\sigma),
\end{align*}
where the last equality is due to Proposition \ref{P-2.16}.
\end{proof}

\begin{thm}\label{T-3.4}
Let $\rho,\sigma\in M_*^+$ and assume that $\rho\ll\sigma$ strongly. Then
$$
\widehat S_f(\rho\|\sigma)=\widehat S'_f(\rho\|\sigma).
$$
\end{thm}

\begin{proof}
In view of Definition \ref{D-2.8} and Lemma \ref{L-3.3}, it suffices to show that
\begin{align}\label{F-3.4}
\widehat S'_f(\rho\|\sigma)
=\lim_{\eps\searrow0}\widehat S'_f(\rho+\eps(\rho+\sigma)\|\sigma+\eps(\rho+\sigma)).
\end{align}
Let $e'$ be the $M'$-support of $\sigma$, and $T=T_{\rho/\sigma}$ be as given in Lemma
\ref{L-3.1}. For every $\eps>0$, since $(1+\eps)\sigma+\eps\rho\ll\sigma$ strongly and
$\sigma+\eps\rho\ll\sigma$ strongly, we have $\overline R_{(1+\eps)\sigma+\eps\rho/\sigma}$
and $\overline R_{\sigma+\eps\rho/\sigma}$ with the polar decompositions
$$
\overline R_{(1+\eps)\sigma+\eps\rho/\sigma}
=V_1T_{(1+\eps)\sigma+\eps\rho/\sigma}^{1/2},\qquad
\overline R_{\sigma+\eps\rho/\sigma}=V_2T_{\sigma+\eps\rho/\sigma}^{1/2},
$$
where $V_1,V_2\in M'$ are partial isometries with $V_1^*V_1=V_2^*V_2=e'$. It is easy to verify
that for every $x\in M$,
\begin{align}
(\eps\sigma+(1+\eps)\rho)(x)
&=\bigl\<(\eps e'+(1+\eps)T)^{1/2}h_\sigma^{1/2},
(\eps e'+(1+\eps)T)^{1/2}xh_\sigma^{1/2}\bigr\>, \label{F-3.5}\\
((1+\eps)\sigma+\eps \rho)(x)
&=\bigl\<((1+\eps)e'+\eps T)^{1/2}h_\sigma^{1/2},
((1+\eps)e'+\eps T)^{1/2}xh_\sigma^{1/2}\bigr\>.
\label{F-3.6}
\end{align}
It follows from \eqref{F-3.6} that $T_{(1+\eps)\sigma+\eps\rho/\sigma}=(1+\eps)e'+\eps T$
so that
$V_1((1+\eps)e'+\eps T)^{1/2}h_\sigma^{1/2}=h_{(1+\eps)\sigma+\eps\rho}^{1/2}$.
Therefore,
\begin{align}\label{F-3.7}
h_\sigma^{1/2}=((1+\eps)e'+\eps T)^{-1/2}V_1^*h_{(1+\eps)\sigma+\eps \rho}^{1/2},
\end{align}
where we note that $(1+\eps)e'+\eps T$ and $e'+\eps T$ have the bounded inverses in $e'Me'$.

Inserting \eqref{F-3.7} into \eqref{F-3.5} gives
\begin{align*}
(\eps\sigma+(1+\eps)\rho)(x)
&=\Bigl\<(\eps e'+(1+\eps)T)^{1/2}((1+\eps)e'+\eps T)^{-1/2}
V_1^*h_{(1+\eps)\sigma+\eps\rho}^{1/2}, \\
&\hskip1cm
(\eps e'+(1+\eps)T)^{1/2}x((1+\eps)e'+\eps T)^{-1/2}
V_1^*h_{(1+\eps)\sigma+\eps\rho}^{1/2}\Bigr\> \\
&=\Bigl\<V_1((\eps e'+(1+\eps)T)((1+\eps)e'+\eps T)^{-1})^{-1/2}
V_1^*h_{(1+\eps)\sigma+\eps\rho}^{1/2}, \\
&\hskip1cm
V_1((\eps e'+(1+\eps)T)((1+\eps)e'+\eps T)^{-1})^{-1/2}
V_1^*xh_{(1+\eps)\sigma+\eps\rho}^{1/2}\Bigr\>,
\end{align*}
which implies that
$$
T_{\eps \sigma+(1+\eps)\rho/(1+\eps)\sigma+\eps \rho}
=V_1(\eps e'+(1+\eps)T)((1+\eps)e'+\eps T)^{-1}V_1^*.
$$
Therefore,
\begin{align*}
&\widehat S'_f(\eps\sigma+(1+\eps)\rho\|(1+\eps)\sigma+\eps\rho) \\
&\quad=\Bigl\<h_{(1+\eps)\sigma+\eps\rho}^{1/2},
f(V_1(\eps e'+(1+\eps)T)((1+\eps)e'+\eps T)^{-1}V_1^*)
h_{(1+\eps)\sigma+\eps\rho}^{1/2}\Bigr\>, \\
&\quad=\Bigl\<V_1((1+\eps)e'+\eps T)^{1/2}h_\sigma^{1/2}, \\
&\qquad\qquad V_1f((\eps e'+(1+\eps)T)((1+\eps)e'+\eps T)^{-1})V_1^*
V_1((1+\eps)e'+\eps T)^{1/2}h_\sigma^{1/2}\Bigr\> \\
&\quad=\Bigl\<((1+\eps)e'+\eps T)^{1/2}h_\sigma^{1/2}, \\
&\qquad\qquad f((\eps e'+(1+\eps)T)((1+\eps)e'+\eps T)^{-1})
((1+\eps)e'+\eps T)^{1/2}h_\sigma^{1/2}\Bigr\> \\
&\quad=\int_0^\infty((1+\eps)+\eps t)
f\biggl({\eps+(1+\eps)t\over(1+\eps)+\eps t}\biggr)
\,d\|E_{\rho/\sigma}(t)h_\sigma^{1/2}\|^2,
\end{align*}
where $(\eps e'+(1+\eps)T)((1+\eps)e'+\eps T)^{-1}$ is a bounded operator and
$E_{\rho/\sigma}(\cdot)$ is the spectral measure of $T$. Hence, to show \eqref{F-3.4}, it
suffices to prove that
\begin{align}\label{F-3.8}
\lim_{\eps\searrow0}\int_0^\infty(1+\eps(1+t))
f\biggl({\eps+(1+\eps)t\over(1+\eps)+\eps t}\biggr)\,d\nu(t)
=\int_0^\infty f(t)\,d\nu(t),
\end{align}
where $d\nu(t):=d\|E_{\rho/\sigma}(t)h_\sigma^{1/2}\|^2$, a finite positive measure on
$[0,\infty)$.

To prove \eqref{F-3.8}, recall \cite[Theorem 8.1]{HMPB} that $f$ admits the integral expression
\begin{align}\label{F-3.9}
f(t)=a+bt+ct^2+\int_{(0,+\infty)}\psi_s(t)\,d\mu(s),
\qquad t\in[0,+\infty),
\end{align}
where $a,b\in\bR$, $c\ge0$,
$$
\psi_s(t):={t\over1+s}-{t\over t+s},\qquad s\in(0,+\infty),\ t\in[0,+\infty),
$$
and $\mu$ is a positive measure on $(0,+\infty)$ satisfying
$\int_{(0,+\infty)}(1+s)^{-2}\,d\mu(s)<+\infty$. Now, let $0<\eps<1/2$ and divide the left-hand
integral in \eqref{F-3.8} into two parts on $[0,5]$ and $(5,+\infty)$. Since
$(1+\eps(1+t))f\Bigl({\eps+(1+\eps)t\over(1+\eps)+\eps t}\Bigr)$ is uniformly bounded for
$\eps\in[0,1/2]$ and $t\in[0,5]$ and
$(1+\eps(1+t))f\Bigl({\eps+(1+\eps)t\over(1+\eps)+\eps t}\Bigr)\to f(t)$ as $\eps\searrow0$ for
every $t\in[0,5]$, the bounded convergence theorem gives
$$
\lim_{\eps\searrow0}\int_{[0,5]}(1+\eps(1+t))
f\biggl({t+\eps(1+t)\over1+\eps(1+t)}\biggr)\,d\nu(t)
=\int_{[0,5]}f(t)\,d\nu(t).
$$
To deal with the integral on $(5,+\infty)$, note that for every $\eps\in(0,1/2)$ and
$t\in(5,+\infty)$,
$$
{d\over d\eps}\,{(t+\eps(1+t))^2\over1+\eps(1+t)}
={(1+t)(t+\eps(1+t))(2-t+\eps(1+t))\over
(1+\eps(1+t))^2}<0,
$$
so that for any $t>5$,
$$
{(t+\eps(1+t))^2\over1+\eps(1+t)}\ \nearrow t^2\ \quad\mbox{as $\eps\searrow0$}.
$$
Moreover, we compute
\begin{align*}
(1+\eps(1+t))
\psi_s\biggl({t+\eps(1+t)\over1+\eps(1+t)}\biggr)
&={t+\eps(1+t)\over1+s}-{t+\eps(1+t)\over
{t+\eps(1+t)\over1+\eps(1+t)}+s} \\
&={t-1\over1+s}\cdot{t+\eps(1+t)\over t+s+\eps(1+t)(1+s)},
\end{align*}
and for any $s\in(0,+\infty)$ and $t>5$,
\begin{align*}
{d\over d\eps}\,{t+\eps(1+t)\over t+s+\eps(1+t)(1+s)}
={(1+t)(1-t)s\over(t+s+\eps(1+t)(1+s))^2}<0,
\end{align*}
so that
$$
0\le(1+\eps(1+t))
\psi_s\biggl({t+\eps(1+t)\over1+\eps(1+t)}\biggr)
\ \nearrow\ \psi_s(t)\quad\mbox{as $\eps\searrow0$}.
$$
The monotone convergence theorem yields that
\begin{align*}
&\int_{(5,+\infty)}(1+\eps(1+t))
f\biggl({t+\eps(1+t)\over1+\eps(1+t)}\biggr)\,d\nu(t) \\
&\quad=(a+(a+b)\eps)\int_{(5,+\infty)}d\nu(t)
+(b+(a+b)\eps)\int_{(5,+\infty)}t\,d\nu(t) \\
&\qquad\quad+c\int_{(5,+\infty)}{(t+\eps(1+t))^2\over
1+\eps(1+t)}\,d\nu(t) \\
&\qquad\quad+\int_{(5,+\infty)}\int_{(0,+\infty)}
(1+\eps(1+t))\psi_s\biggl({t+\eps(1+t)\over
1+\eps(1+t)}\biggr)\,d\mu(s)\,d\nu(t)
\end{align*}
converges as $\eps\searrow0$ to
\begin{align*}
&a\int_{(5,+\infty)}d\nu(t)+b\int_{(5,+\infty)}t\,d\nu(t)
+c\int_{(5,+\infty)}t^2\,d\nu(t)
+\int_{(5,+\infty)}\int_{(0,+\infty)}\psi_s(t)\,d\mu(s)\,d\nu(t) \\
&\qquad=\int_{(5,+\infty)}\biggl(a+bt+ct^2
+\int_{(0,+\infty)}\psi_s(t)\,d\mu(s)\biggr)\,d\nu(t)
=\int_{(5,+\infty)}f(t)\,d\nu(t).
\end{align*}
Hence \eqref{F-3.8} follows.
\end{proof}

\begin{example}\label{E-3.5}\rm
In \cite{BS} Belavkin and Staszewski introduced a kind of relative entropy for states on a
$C^*$-algebra. Here we restrict to the von Neumann algebra setting. Let $\rho,\sigma\in M_*^+$
and assume that $\rho\ll\sigma$ strongly. Let $R_{\rho/\sigma}$ and $T=T_{\rho/\sigma}$ be
as in \eqref{F-3.1} and Lemma \ref{L-3.1} so that $\overline R_{\rho/\sigma}=VT^{1/2}$. Note
that
$$
T^{1/2}xh_\sigma^{1/2}=V^*xh_\rho^{1/2}=xV^*h_\rho^{1/2},\qquad x\in M.
$$
Hence, the vector $\xi$ and the positive self-adjoint operator $\overline{\rho(\xi)}$ in
\cite[(3.2)]{BS} are
$$
\xi=V^*h_\rho^{1/2}=T^{1/2}h_\sigma^{1/2},\qquad\overline{\rho(\xi)}=T^{1/2}.
$$
Therefore, {\it Belavkin and Staszewski's relative entropy} $D_\BS(\rho\|\sigma)$ is given as
\begin{align*}
D_\BS(\rho\|\sigma)&=\<T^{1/2}h_\sigma^{1/2},(\log T)T^{1/2}h_\sigma^{1/2}\> \\
&=\lim_{\delta\searrow0}\int_\delta^{\delta^{-1}}\log t
\,d\|E(t)T^{1/2}h_\sigma^{1/2}\|^2 \\
&=\lim_{\delta\searrow0}\int_\delta^{\delta^{-1}}\eta(t)
\,d\|E(t)h_\sigma^{1/2}\|^2 \\
&=\<h_\sigma^{1/2},\eta(T)h_\sigma^{1/2}\>
=\widehat S'_\eta(\rho\|\sigma)=\widehat S_\eta(\rho\|\sigma).
\end{align*}
where $T=\int_0^\infty t\,dE(t)$ is the spectral decomposition and $\eta(t):=t\log t$. In this
way, when $\rho\ll\sigma$ strongly, $D_\BS(\rho\|\sigma)$ is realized as the maximal
$f$-divergence $\widehat S_f(\rho\|\sigma)$ with $f=\eta$. Thus we may and do define
$D_\BS(\rho\|\sigma):=\widehat S_\eta(\rho\|\sigma)$ for arbitrary $\rho,\sigma\in M_*^+$.
\end{example}

\section{General integral formula}

We modify the arguments in the previous section to show the following:

\begin{prop}\label{P-4.1}
Let $\rho,\sigma\in M_*^+$.
\begin{itemize}
\item[(1)] If $f(0^+)<+\infty$, then
$$
\widehat S_f(\rho\|\sigma)=\lim_{\eps\searrow0}\widehat S_f(\rho\|\sigma+\eps\rho).
$$
\item[(2)] If $f'(+\infty)<+\infty$, then
$$
\widehat S_f(\rho\|\sigma)=\lim_{\eps\searrow0}\widehat S_f(\rho+\eps\sigma\|\sigma).
$$
\end{itemize}
\end{prop}

\begin{proof}
(1)\enspace
Assume that $f(0^+)<+\infty$, and extend $f$ to $[0,+\infty)$ by $f(0)=f(0^+)$. Set
$\eta:=\rho+\sigma$ and let $\eps>0$. Since $\rho$, $\sigma$ and $\sigma+\eps\rho$ are all
dominated by $\eta$, we have the three (bounded) positive self-adjoint operators
$T_1:=T_{\rho/\eta}$, $T_2:=T_{\sigma/\eta}$ and $T_3:=T_{\sigma+\eps\rho/\eta}$ as follows:
\begin{align}\label{F-4.1}
\overline R_{\rho/\eta}=V_1T_1^{1/2},\qquad
\overline R_{\sigma/\eta}=V_2T_2^{1/2},\qquad
\overline R_{\sigma+\eps\rho/\eta}=V_3T_3^{1/2},
\end{align}
where $V_k$'s are partial isometries in $e'M'e'$ with $V_3^*V_3=e'$, where $e'$ is the
projection onto $\overline{Mh_\eta^{1/2}}$. Since
$$
\rho(x)=\bigl\<T_1h_\eta^{1/2},xh_\eta^{1/2}\bigr\>,\qquad
\sigma(x)=\bigl\<T_2h_\eta^{1/2},xh_\eta^{1/2}\bigr\>,\qquad x\in M,
$$
we see that $T_1+T_2=e'$. Moreover, since
$$
(\sigma+\eps\rho)(x)=\bigl\<(T_2+\eps T_1)h_\eta^{1/2},xh_\eta^{1/2}\bigr\>,\qquad x\in M,
$$
we have $\overline R_{\sigma+\eps\rho}=V_3(T_2+\eps T_1)^{1/2}$ so that
$V_3(T_2+\eps T_1)^{1/2}h_\eta^{1/2}=h_{\sigma+\eps\rho}^{1/2}$.
We find that
\begin{align*}
(\rho+\eps\sigma)(x)
&=\bigl\<(T_1+\eps T_2)h_\eta^{1/2},xh_\eta^{1/2}\bigr\> \\
&=\bigl\<(T_1+\eps T_2)(T_2+\eps T_1)^{-1/2}V_3^*h_{\sigma+\eps\rho}^{1/2},
x(T_2+\eps T_1)^{-1/2}V_3^*h_{\sigma+\eps\rho}^{1/2}\bigr\> \\
&=\bigl\<V_3(T_1+\eps T_2)(T_2+\eps T_1)^{-1}V_3^*h_{\sigma+\eps\rho}^{1/2},
xh_{\sigma+\eps\rho}^{1/2}\bigr\>,\qquad x\in M,
\end{align*}
which implies that
$$
T_{\rho+\eps\sigma/\sigma+\eps\rho}=V_3(T_1+\eps T_2)(T_2+\eps T_1)^{-1}V_3^*.
$$
Therefore,
\begin{align*}
&\widehat S'_f(\rho+\eps\sigma\|\sigma+\eps\rho) \\
&\quad=\bigl\<h_{\sigma+\eps\rho}^{1/2},
f(V_3(T_1+\eps T_2)(T_2+\eps T_1)^{-1}V_3^*)h_{\sigma+\eps\rho}^{1/2}\bigr\> \\
&\quad=\bigl\<V_3(T_2+\eps T_1)^{1/2}h_\eta^{1/2},
V_3f((T_1+\eps T_2)(T_2+\eps T_1)^{-1})V_3^*V_3(T_2+\eps T_1)^{1/2}h_\eta^{1/2}\bigr\> \\
&\quad=\bigl\<h_\eta^{1/2},
(T_2+\eps T_1)f((T_1+\eps T_2)(T_2+\eps T_1)^{-1})h_\eta^{1/2}\bigr\>.
\end{align*}
Now, since $0\le T_1\le e'$ and $T_2=e'-T_1$, taking the spectral decomposition
$T_1=\int_0^1t\,dE_1(t)$ with $\int_0^1dE_1(t)=e'$, one can write
\begin{align}\label{F-4.2}
\widehat S'_f(\rho+\eps\sigma\|\sigma+\eps\rho)
=\int_0^1(1-t+\eps t)f\biggl({t+\eps(1-t)\over1-t+\eps t}\biggr)\,d\nu(t),
\end{align}
where $d\nu(t):=d\|E_1(t)h_\eta^{1/2}\|^2$, a finite positive measure on $[0,1]$.

Similarly, one has
$$
\rho(x)=\bigl\<V_3T_1(T_2+\eps T_1)^{-1}V_3^*h_{\sigma+\eps\rho}^{1/2},
xh_{\sigma+\eps\rho}^{1/2}\bigr\>,\qquad x\in M,
$$
so that
$$
T_{\rho/\sigma+\eps\rho}=V_3T_1(T_2+\eps T_1)^{-1}V_3^*.
$$
Therefore,
\begin{align}
\widehat S'_f(\rho\|\sigma+\eps\rho)
&=\<h_\eta^{1/2},(T_2+\eps T_1)f(T_1(T_2+\eps T_1)^{-1})h_\eta^{1/2}\> \nonumber\\
&=\int_0^1(1-t+\eps t)f\biggl({t\over1-t+\eps t}\biggr)\,d\nu(t). \label{F-4.3}
\end{align}

By Lemma \ref{L-3.3} one has
\begin{align}\label{F-4.4}
\widehat S_f(\rho+\eps\sigma\|\sigma+\eps\rho)
=\widehat S'_f(\rho+\eps\sigma\|\sigma+\eps\rho),\qquad
\widehat S_f(\rho\|\sigma+\eps\rho)=\widehat S'_f(\rho\|\sigma+\eps\rho).
\end{align}
Furthermore, by Definition \ref{D-2.8} one sees that
\begin{align}
\widehat S_f(\rho\|\sigma)
&=\lim_{\eps\searrow0}\widehat S_f(\rho+\eps(\rho+\sigma)\|\sigma+\eps(\rho+\sigma))
\nonumber \\
&=\lim_{\eps\searrow0}\widehat S_f((1+\eps)\rho+\eps\sigma\|(1+\eps)\sigma+\eps\rho)
\nonumber \\
&=\lim_{\eps\searrow0}\widehat S_f\biggl(\rho+{\eps\over1+\eps}\,\sigma\|
\sigma+{\eps\over1+\eps}\,\rho\biggr) \nonumber\\
&=\lim_{\eps\searrow0}\widehat S_f(\rho+\eps\sigma\|\sigma+\eps\rho). \label{F-4.5}
\end{align}
From \eqref{F-4.2}--\eqref{F-4.5} it suffices to prove that the integrals in \eqref{F-4.2} and
\eqref{F-4.3} have the same limit
as $\eps\searrow0$. For this we may prove that
\begin{align}
\lim_{\eps\searrow0}\int_0^1(1-t+\eps t)
f\biggl({t+\eps(1-t)\over1-t+\eps t}\biggr)\,d\nu(t)
&=\int_0^1(1-t)f\biggl({t\over1-t}\biggr)\,d\nu(t), \label{F-4.6}\\
\lim_{\eps\searrow0}\int_0^1(1-t+\eps t)
f\biggl({t\over1-t+\eps t}\biggr)\,d\nu(t)
&=\int_0^1(1-t)f\biggl({t\over1-t}\biggr)\,d\nu(t), \label{F-4.7}
\end{align}
where $(1-t)f(t/(1-t))$ at $t=1$ is understood as
$\lim_{t\nearrow1}(1-t)f\bigl({t\over1-t}\bigr)=f'(+\infty)$. Let us transfer the proofs of
\eqref{F-4.6} and \eqref{F-4.7} into Appendix A, which are more or less similar to that of
\eqref{F-3.8}.

(2) is immediate from (1) and Proposition \ref{P-2.10}.
\end{proof}

In the proof of Theorem \ref{T-3.4} we used the integral expression of an operator convex
function $f$ on $(0,+\infty)$ satisfying $f(0^+)<+\infty$. For a general operator convex
function $f$ on $(0,+\infty)$, recall \cite{LR} (see also \cite[Theorem 5.1]{FHR}) that $f$
has an integral expression
\begin{align}\label{F-4.8}
f(t)=a+b(t-1)+c(t-1)^2+d\,{(t-1)^2\over t}+\int_{(0,+\infty)}{(t-1)^2\over t+s}\,d\mu(s),
\quad t\in(0,+\infty).
\end{align}
where $a,b\in\bR$, $c,d\ge0$ and $\mu$ is a positive measure on $[0,+\infty)$ with
$\int_{[0,+\infty)}(1+s)^{-1}\,d\mu(s)<+\infty$, and moreover $a,b,c,d$ and $\mu$ are uniquely
determined.

Based on the arguments in the proof of Proposition \ref{P-4.1}, we next present a general
integral formula of $\widehat S_f(\rho\|\sigma)$, which can be the second definition of the
maximal $f$-divergences.

\begin{thm}\label{T-4.2}
For every $\rho,\sigma\in M_*^+$ let $T_{\rho/\rho+\sigma}=\int_0^1t\,dE_{\rho/\rho+\sigma}(t)$
be the spectral decomposition. Then for every operator convex
function $f$ on $(0,+\infty)$,
\begin{align}\label{F-4.9}
\widehat S_f(\rho\|\sigma)=\int_0^1(1-t)f\biggl({t\over1-t}\biggr)
\,d\|E_{\rho/\rho+\sigma}(t)h_{\rho+\sigma}^{1/2}\|^2,
\end{align}
where $(1-t)f\bigl({t\over1-t}\bigr)$ is understood as $f(0^+)$ for $t=0$ and $f'(+\infty)$
for $t=1$.
\end{thm}

\begin{proof}
In view of the integral expression in \eqref{F-4.8}, we can write $f=f_1+f_2$ with operator
convex functions $f_1,f_2$ on $(0,+\infty)$ such that $f_1(0^+)<+\infty$ and
$f_2'(+\infty)<+\infty$ and so $\widetilde f_2(0^+)<+\infty$. In fact, we may define
\begin{align*}
f_1(t)&:=a+b(t-1)+c(t-1)^2+\int_{[1,+\infty)}{(t-1)^2\over t+s}\,d\mu(s), \\
f_2(t)&:=d\,{(t-1)^2\over t}+\int_{(0,1)}{(t-1)^2\over t+s}\,d\mu(s).
\end{align*}
We then have
$$
\widehat S_f(\rho\|\sigma)
=\widehat S_{f_1}(\rho\|\sigma)+\widehat S_{f_2}(\rho\|\sigma)
=\widehat S_{f_1}(\rho\|\sigma)+\widehat S_{\widetilde f_2}(\sigma\|\rho).
$$
It follows from \eqref{F-4.5}, \eqref{F-4.2} and \eqref{F-4.6} in the proof of Proposition
\ref{P-4.1}\,(1) that
\begin{align}\label{F-4.10}
\widehat S_{f_1}(\rho\|\sigma)=\int_0^1(1-t)f_1\biggl({t\over1-t}\biggr)
\,d\|E_{\rho/\rho+\sigma}(t)h_{\rho+\sigma}^{1/2}\|^2,
\end{align}
and similarly, with $\rho,\sigma$ interchanged,
\begin{align}\label{F-4.11}
\widehat S_{\widetilde f_2}(\sigma\|\rho)
=\int_0^1(1-t)\widetilde f_2\biggl({t\over1-t}\biggr)
\,d\|E_{\sigma/\rho+\sigma}(t)h_{\rho+\sigma}^{1/2}\|^2,
\end{align}
where $T_{\sigma/\rho+\sigma}=\int_0^1t\,dE_{\sigma/\rho+\sigma}(t)$ is the spectral
decomposition. Since $T_{\sigma/\rho+\sigma}=e'-T_{\rho/\rho+\sigma}$, we find that
$E_{\sigma/\rho+\sigma}([0,t])=E_{\rho/\rho+\sigma}([1-t,1])$ for all $t\in[0,1]$.
Applying this to \eqref{F-4.11} gives
\begin{align}
\widehat S_{\widetilde f_2}(\sigma\|\rho)
&=\int_0^1t\widetilde f_2\biggl({1-t\over t}\biggr)
\,d\|E_{\rho/\rho+\sigma}(t)h_{\rho+\sigma}^{1/2}\|^2 \nonumber\\
&=\int_0^1(1-t)f_2\biggl({t\over1-t}\biggr)
\,d\|E_{\rho/\rho+\sigma}(t)h_{\rho+\sigma}^{1/2}\|^2. \label{F-4.12}
\end{align}
Hence \eqref{F-4.9} follows by adding \eqref{F-4.10} and \eqref{F-4.12}.
\end{proof}

\begin{cor}\label{C-4.4}
If $f(0^+)<+\infty$ and $f'(+\infty)<+\infty$, then $\widehat S_f(\rho\|\sigma)$
is finite for every $\rho,\sigma\in M_*^+$.
\end{cor}

\begin{proof}
If $f(0^+)<+\infty$ and $f'(+\infty)<+\infty$, then the function
$(1-t)f\bigl({t\over1-t}\bigr)$ is bounded on $[0,1]$. Hence the result is obvious from
expression \eqref{F-4.9}.
\end{proof}

\begin{remark}\label{R-4.4}\rm
Similarly to the definition of the standard $f$-divergence $S_f(\rho\|\sigma)$ in
\cite[(2.6)]{Hi1}, one can write \eqref{F-4.9} as sum of three terms
\begin{align}
\widehat S_f(\rho\|\sigma)
&=\int_{(0,1)}(1-t)f\biggl({t\over1-t}\biggr)
\,d\|E_{\rho/\rho+\sigma}(t)h_{\rho+\sigma}^{1/2}\|^2 \nonumber\\
&\quad+f(0^+)\bigl\<h_{\rho+\sigma}^{1/2},
E_{\rho/\rho+\sigma}(\{0\})h_{\rho+\sigma}^{1/2}\bigr\>
+f'(+\infty)\bigl\<h_{\rho+\sigma}^{1/2},
E_{\rho/\rho+\sigma}(\{1\})h_{\rho+\sigma}^{1/2}\bigr\>, \label{F-4.13}
\end{align}
where $E_{\rho/\rho+\sigma}(\{0\})$ and $E_{\rho/\rho+\sigma}(\{1\})$ are the spectral
projections of $T_{\rho/\rho+\sigma}$ for $\{0\}$ and $\{1\}$. One might expect that the above
boundary terms with $f(0^+)$ and $f'(+\infty)$ are equal to the corresponding terms
$f(0^+)\sigma(1-s(\rho))$ and $f'(+\infty)\rho(1-s(\sigma))$, respectively, in
\cite[(2.6)]{Hi1}. But it is not true, as will explicitly be seen in Example \ref{E-4.5} below.
Instead, it is not difficult to find that
\begin{align*}
\bigl\<h_{\rho+\sigma}^{1/2},
V_1E_{\rho/\rho+\sigma}(\{0\})V_1^*h_{\rho+\sigma}^{1/2}\bigr\>
&=\sigma(1-s(\rho)), \\
\bigl\<h_{\rho+\sigma}^{1/2},
V_2E_{\rho/\rho+\sigma}(\{1\})V_2^*h_{\rho+\sigma}^{1/2}\bigr\>
&=\rho(1-s(\sigma)),
\end{align*}
where $V_1,V_2$ are partial isometries in \eqref{F-4.1}.
\end{remark}

\begin{example}\label{E-4.5}\rm
Although we noted in Example \ref{E-2.12} that the definition of $\widehat S_f(\rho\|\sigma)$
in \eqref{F-2.4} coincides with \cite[Definition 3.21]{HM} in the finite-dimensional case, we
examine formula \eqref{F-4.9} in the matrix case $M=\bM_d$. For $\rho,\sigma\in\bM_d^+$, since
$R_{\rho/\rho+\sigma}(X(\rho+\sigma)^{1/2})=X\rho^{1/2}$, one has
$R_{\rho/\rho+\sigma}X=X(\rho+\sigma)^{-1/2}\rho^{1/2}$ and hence
$R_{\rho/\rho+\sigma}^*X=X\rho^{1/2}(\rho+\sigma)^{-1/2}$, where $(\rho+\sigma)^{-1/2}$ is
defined in the sense of generalized inverse. Therefore,
$$
T_{\rho/\rho+\sigma}X=R_{\rho/\rho+\sigma}^*R_{\rho/\rho+\sigma}X
=X(\rho+\sigma)^{-1/2}\rho(\rho+\sigma)^{-1/2},
$$
so that $T_{\rho/\rho+\sigma}=R_{(\rho+\sigma)^{-1/2}\rho(\rho+\sigma)^{-1/2}}$ and similarly
$T_{\sigma/\rho+\sigma}=R_{(\rho+\sigma)^{-1/2}\sigma(\rho+\sigma)^{-1/2}}$, where $R_A$ is
the right multiplication by $A$ on $\bM_d$. Note that
$T_{\rho/\rho+\sigma}+T_{\sigma/\rho+\sigma}=e'$ is the right multiplication of the support
projection of $\rho+\sigma$. Here, for simplicity, assume that $\sigma$ is invertible,
and let $\sigma^{1/2}(\rho+\sigma)^{-1/2}= VQ^{1/2}$ be the polar decomposition where
$Q:=(\rho+\sigma)^{-1/2}\sigma(\rho+\sigma)^{-1/2}$. Then formula \eqref{F-4.9} is written
as
\begin{align}
\widehat S_f(\rho\|\sigma)
&=\bigl\<(\rho+\sigma)^{1/2},R_Qf(R_{I-Q}R_Q^{-1})(\rho+\sigma)^{1/2}\bigr\> \nonumber\\
&=\bigl\<R_Q^{1/2}(\rho+\sigma)^{1/2},R_{f(Q^{-1}-I)}R_Q^{1/2}(\rho+\sigma)^{1/2}\bigr\>
\nonumber\\
&=\Tr Q^{1/2}(\rho+\sigma)Q^{1/2}f(Q^{-1}-I)=\Tr V^*\sigma Vf(Q^{-1}-I) \nonumber\\
&=\Tr\sigma f(V(Q^{-1}-I)V^*)=\Tr\sigma f((VQV^*)^{-1}-I) \nonumber\\
&=\Tr\sigma f((\sigma^{1/2}(\rho+\sigma)^{-1}\sigma^{1/2})^{-1}-I)
=\Tr\sigma f(\sigma^{-1/2}\rho\sigma^{-1/2}), \label{F-4.14}
\end{align}
which coincides with \cite[Definition 3.21]{HM}, as mentioned in Definition \ref{D-2.8}.

When $\rho,\sigma$ are not invertible, the two boundary terms in \eqref{F-4.13} are
\begin{align}
f(0^+)&\<(\rho+\sigma)^{1/2},R_{E_0}(\rho+\sigma)^{1/2}\>, \label{F-4.15}\\
f'(+\infty)&\<(\rho+\sigma)^{1/2},R_{E_1}(\rho+\sigma)^{1/2}\>, \label{F-4.16}
\end{align} 
where $E_0,E_1$ are the spectral projections of
$(\rho+\sigma)^{-1/2}\sigma(\rho+\sigma)^{-1/2}$ for the eigenvalues $0,1$, respectively.
For example, consider the $2\times2$ matrix case where
$\rho:=\begin{bmatrix}{3\over2}&0\\0&0\end{bmatrix}$ and
$\sigma:=\begin{bmatrix}1&1\\1&1\end{bmatrix}$. By direct computations we find that
\eqref{F-4.15} and \eqref{F-4.16} are equal to
$$
2f(0^+),\qquad{3\over2}\,f'(+\infty),
$$
respectively. On the other hand, the corresponding terms in the definition of
$S_f(\rho\|\sigma)$ in \cite[(2.6)]{Hi1} are equal to
$$
f(0^+)\sigma(1-s(\rho))=f(0^+),\qquad
f'(+\infty)\rho(1-s(\sigma))={3\over4}\,f'(+\infty),
$$
in this case. Thus, the two boundary terms with $f(0^+)$ and $f'(+\infty)$ for
$S_f(\rho\|\sigma)$ and $\widehat S_f(\rho\|\sigma)$ are different each other.
\end{example}

\section{Lower semicontinuity and martingale convergence}

For each $n\in\bN$, as in \cite[(3.6)]{Hi1}, we consider the approximation of $f$ in
\eqref{F-4.8} with integral on the cut-off interval $[1/n,n]$, that is,
\begin{align}
f_n(t)&:=a+b(t-1)+c\,{n(t-1)^2\over t+n}+d\,{(t-1)^2\over t+(1/n)} \nonumber\\
&\qquad+\int_{[1/n,n]}{(t-1)^2\over t+s}\,d\mu(s),\qquad d\in(0,+\infty). \label{F-5.1}
\end{align}
The following lemma is \cite[Lemma 3.1]{Hi1}.

\begin{lemma}\label{L-5.1}
For each $n\in\bN$, $f_n$ is operator convex on $(0,+\infty)$, $f_n(0^+)<+\infty$,
$f_n'(+\infty)<+\infty$ and
$$
f_n(0^+)\,\nearrow\,f(0^+),\qquad f_n'(+\infty)\,\nearrow\,f'(+\infty),\qquad
f_n(t)\,\nearrow\,f(t)
$$
for all $t\in(0,+\infty)$ as $n\to\infty$.
\end{lemma}

\begin{lemma}\label{L-5.2}
For every $\rho,\sigma\in M_*^+$,
$$
\widehat S_f(\rho\|\sigma)=\lim_{n\to\infty}\widehat S_{f_n}(\rho\|\sigma)
\quad\mbox{increasingly}.
$$
Hence $\widehat S_f(\rho\|\sigma)=\sup_{n\ge1}\widehat S_{f_n}(\rho\|\sigma)$.
\end{lemma}

\begin{proof}
By Lemma \ref{L-5.1}, as $n\to\infty$,
$$
(1-t)f_n\biggl({t\over 1-t}\biggr)\ \nearrow\ (1-t)f\biggl({t\over1-t}\biggr),\qquad
t\in[0,1].
$$
Hence the result follows from the monotone convergence theorem applied to the integral formula
in \eqref{F-4.9} for $f_n$ and $f$.
\end{proof}

In addition to Lemma \ref{L-5.1}, it is readily verified that
$$
\lim_{t\searrow0}f_n'(t)=\lim_{t\searrow0}{f_n(t)-f_n(0^+)\over t}>-\infty.
$$
So, to prove the joint lower semicontinuity of $(\rho,\sigma)\mapsto\widehat S_f(\rho\|\sigma)$,
we may and do assume that $f(0^+)<+\infty$, $f'(+\infty)<+\infty$ and
$\lim_{t\searrow0}f'(t)>-\infty$. Such an operator monotone function $f$ on
$[0,+\infty)$, with $f(0)=f(0^+)$, has the integral expression
\begin{align}\label{F-5.2}
f(t)=a+bt+\int_{(0,+\infty)}{t^2\over t+s}\,d\nu(s),\qquad t\in[0,+\infty),
\end{align}
where $a,b\in\bR$ and $\nu$ is a finite positive measure on $(0,\infty)$. Indeed, the function
$g(t):=(f(t)-f(0^+))/t$ is operator monotone on $[0,+\infty)$ by \cite[Theorem 2.4]{HP}, so
$g$ has the integral expression
$$
g(t)=b+ct+\int_{(0,+\infty)}{t\over t+s}\,d\nu(s),\qquad t\in[0,+\infty),
$$
where $b\in\bR$, $c\ge0$ and $\nu$ is a positive measure on $(0,+\infty)$, see
\cite[(V.53)]{Bh} (also \cite[Theorem 2.7.11]{Hi}). Since $g(+\infty)=f'(+\infty)<+\infty$, it
must follow that $c=0$ and $\nu$ is a finite measure. Hence $f$ has the expression in
\eqref{F-5.2}. In view of \eqref{F-2.6} we may and do furthermore assume that
\begin{align}\label{F-5.3}
f(t)=\int_{(0,+\infty)}{t^2\over t+s}\,d\nu(s),\qquad t\in[0,+\infty),
\end{align}
where $\nu$ is as above.

\begin{lemma}\label{L-5.3}
Let $f$ be given in \eqref{F-5.3}. Then for every $\rho,\sigma\in M_*^+$,
$$
\widehat S_f(\rho\|\sigma)=\lim_{\eps\searrow0}\widehat S_f(\rho\|\sigma+\eps\rho)
\quad\mbox{increasingly}.
$$
Hence $\widehat S_f(\rho\|\sigma)=\sup_{\eps>0}\widehat S_f(\rho\|\sigma+\eps\rho)$.
\end{lemma}

\begin{proof}
Since $f(0^+)<+\infty$, the convergence is in Proposition \ref{P-4.1}\,(1). So we need to show
that $0<\eps\mapsto\widehat S_f(\rho\|\sigma+\eps\rho)$ is decreasing. Let $\eta:=\rho+\sigma$.
By Proposition \ref{P-2.10} we have 
$$
\widehat S_f(\rho\|\sigma+\eps\rho)=\widehat S_{\widetilde f}(\sigma+\eps\rho\|\rho)
=\lim_{\delta\searrow0}\widehat S_{\widetilde f}(\sigma+\eps\rho+\delta\eta\|\rho+\delta\eta).
$$
Hence it suffices to show that for $0<\eps_1<\eps_2$ and $\delta>0$,
$$
\widehat S_{\widetilde f}(\sigma+\eps_1\rho+\delta\eta\|\rho+\delta\eta)
\ge\widehat S_{\widetilde f}(\sigma+\eps_2\rho+\delta\eta\|\rho+\delta\eta).
$$
Set $\sigma_i:=\sigma+\eps_i\rho+\delta\eta$ and $\omega:=\rho+\delta\eta$; then
$\sigma_1\le\sigma_2$ and $\sigma_1\sim\sigma_2\sim\omega$. Let $e:=s(\omega)$. By Lemma
\ref{L-2.1} there is an $A\in eMe$ such that $h_{\sigma_2}^{1/2}=Ah_\omega^{1/2}$. Also there
is a $B\in eMe$ such that $\|B\|\le1$ and $h_{\sigma_1}^{1/2}=Bh_{\sigma_2}^{1/2}$. Since
$h_{\sigma_1}^{1/2}=BAh_\omega^{1/2}$, one has
$$
\widehat S_{\widetilde f}(\sigma_1\|\omega)=\omega(\widetilde f(A^*B^*BA)),\qquad
\widehat S_{\widetilde f}(\sigma_2\|\omega)=\omega(\widetilde f(A^*A)).
$$
From \eqref{F-5.3} the function $\widetilde f$ is expressed as
$$
\widetilde f(t)=tf(t^{-1})=\int_{(0,+\infty)}{1\over1+st}\,d\nu(s),\qquad t\in[0,+\infty),
$$
which is operator monotone decreasing on $[0,+\infty)$. Since $A^*B^*BA\le A^*A$, it follows
that $f(A^*B^*BA)\ge f(A^*A)$ and hence $\widehat S_{\widetilde f}(\sigma_1\|\omega)\ge
\widehat S_{\widetilde f}(\sigma_2\|\omega)$, as desired.
\end{proof}

\begin{lemma}\label{L-5.4}
Let $T,T_n$ ($n\in\bN$) be positive bounded linear operators on a Hilbert space $\cH$. Assume
that $\sup_n\|T_n\|<+\infty$ and $T_n\to T$ in the weak operator topology. Then for any
operator convex function $f$ on $[0,+\infty)$ with $f(0)=0$ and any $\xi\in\cH$,
$$
\<\xi,f(T)\xi\>\le\liminf_{n\to\infty}\<\xi,f(T_n)\xi\>.
$$
\end{lemma}

\begin{proof}
Let $\{E_\alpha\}$ be a net of finite-dimensional orthogonal projections on $\cH$ such that
$E_\alpha\nearrow I$. Since $\|E_\alpha T_nE_\alpha-E_\alpha TE_\alpha\|\to0$ as $n\to\infty$,
$\|f(E_\alpha T_nE_\alpha)-f(E_\alpha TE_\alpha)\|\to0$ as $n\to\infty$ for every $\alpha$. By
\cite[Theorem 2.5]{Ch} applied to the map $\Phi(A)=E_\alpha AE_\alpha+(I-E_\alpha)A(I-E_\alpha)$
on $\cB(\cH)$, we have
\begin{align}\label{F-5.4}
f(E_\alpha T_nE_\alpha)\le E_\alpha f(T_n)E_\alpha.
\end{align}
Therefore,
\begin{align}\label{F-5.5}
\<\xi,f(E_\alpha TE_\alpha)\xi\>&=\lim_{n\to\infty}\<\xi,f(E_\alpha T_nE_\alpha)\xi\>
\le\liminf_{n\to\infty}\<\xi,E_\alpha f(T_n)E_\alpha\xi\>.
\end{align}
Since $E_\alpha TE_\alpha\to T$ in the strong operator topology and the continuous functional
calculus is continuous with respect to the strong operator topology (see, e.g.,
\cite[Theorem A.2]{St}), the left-hand side of \eqref{F-5.5} converges to $\<\xi,f(T)\xi\>$. On
the other hand, since $K:=\sup_n\|f(T_n)\|<+\infty$, note that
\begin{align}
&|\<E_\alpha\xi,f(T_n)E_\alpha\xi\>-\<\xi,f(T_n)\xi\>| \nonumber\\
&\quad\le|\<E_\alpha\xi-\xi,f(T_n)E_\alpha\xi\>|+|\<\xi,f(T_n)(E_\alpha\xi-\xi)\>|
\le2K\|\xi\|\cdot\|E_\alpha\xi-\xi\|, \label{F-5.6}
\end{align}
so $\<E_\alpha\xi,f(T_n)E_\alpha\xi\>$ converges to $\<\xi,f(T_n)\xi\>$ as $\alpha\to``\infty"$ uniformly for
$n$. This implies that the right-hand side of \eqref{F-5.5} converges to
$\liminf_{n\to\infty}\<\xi,f(T_n)\xi\>$ as $\alpha\to``\infty"$. Hence the result follows.
\end{proof}

We are now in a position to prove the joint lower semicontinuity.

\begin{thm}\label{T-5.5}
The function $(\rho,\sigma)\in M_*^+\times M_*^+\mapsto\widehat S_f(\rho\|\sigma)$ is jointly
lower semicontinuous in the norm topology.
\end{thm}

\begin{proof}
By the argument above Lemma \ref{L-5.3}, we may assume that $f$ is given in \eqref{F-5.3}. By
Lemma \ref{L-5.3} it suffices to prove that
$(\rho,\sigma)\in M_*^+\times M_*^+\mapsto \widehat S_f(\rho\|\sigma+\eps\rho)$ is continuous
in the norm topology for any $\eps>0$. Let $\rho_n,\rho,\sigma_n,\sigma\in M_*^+$, $n\in\bN$,
be such that $\|\rho_n-\rho\|\to0$ and $\|\sigma_n-\sigma\|\to0$. Let
$\eta_n:=\sigma_n+\eps\rho_n$ and $\eta:=\sigma+\eps\rho$; then $\rho_n\le\eps^{-1}\eta_n$,
$\rho\le\eps^{-1}\eta$ and $\|\eta_n-\eta\|\to0$. By Lemma \ref{L-2.1} we have
$A\in s(\eta)Ms(\eta)$ and $A_n\in s(\eta_n)Ms(\eta_n)$ such that $h_\rho^{1/2}=Ah_\eta^{1/2}$
and $h_{\rho_n}^{1/2}=A_nh_{\eta_n}^{1/2}$ and $\|A\|,\|A_n\|\le\eps^{-1/2}$. Let $x,y\in M$.
Note that
$$
\|h_{\eta_n}^{1/2}x-h_\eta^{1/2}x\|_2
\le\|h_{\eta_n}^{1/2}-h_\eta^{1/2}\|_2\|x\|\le\|\eta_n-\eta\|^{1/2}\|x\|
$$
thanks to \cite[Lemma 2.10\,(2)]{Ha} and similarly
$\|h_{\rho_n}^{1/2}x-h_\rho^{1/2}x\|_2\le\|\rho_n-\rho\|^{1/2}\|x\|$. Hence,
$$
\<h_{\rho_n}^{1/2}x,h_{\rho_n}^{1/2}y\>\ \longrightarrow\ \<h_\rho^{1/2}x,h_\rho^{1/2}y\>
\quad\mbox{as $n\to\infty$}.
$$
Since $A(h_\eta^{1/2}x)=h_\rho^{1/2}x$ and $A_n(h_{\eta_n}^{1/2}x)=h_{\rho_n}^{1/2}x$, one can
estimate
\begin{align*}
&|\<h_\eta^{1/2}x,A_n^*A_n(h_\eta^{1/2}y)\>-\<h_\eta^{1/2}x,A^*A(h_\eta^{1/2}y)\>| \\
&\quad\le|\<h_\eta^{1/2}x-h_{\eta_n}^{1/2}x,A_n^*A_n(h_\eta^{1/2}y)\>|
+|\<h_{\eta_n}^{1/2}x,A_n^*A_n(h_\eta^{1/2}y-h_{\eta_n}^{1/2}y)\>| \\
&\quad\quad
+|\<h_{\eta_n}^{1/2}x,A_n^*A_n(h_{\eta_n}^{1/2}y)\>-\<h_\eta^{1/2}x,A^*A(h_\eta^{1/2}y)\>| \\
&\quad\le\|h_\eta^{1/2}x-h_{\eta_n}^{1/2}x\|_2\|A_n^*A_n(h_\eta^{1/2}y)\|_2
+\|h_{\eta_n}^{1/2}x\|_2\|A_n^*A_n(h_\eta^{1/2}y-h_{\eta_n}^{1/2}y)\|_2 \\
&\quad\quad
+|\<A_n(h_{\eta_n}^{1/2}x),A_n(h_{\eta_n}^{1/2}y)\>-\<A(h_\eta^{1/2}x),A(h_\eta^{1/2}y)\>| \\
&\quad\le\eps^{-1}\|\eta_n-\eta\|^{1/2}(\|\eta\|^{1/2}+\|\eta_n\|^{1/2})\|x\|\,\|y\|
+|\<h_{\rho_n}^{1/2}x,h_{\rho_n}^{1/2}y\>-\<h_\rho^{1/2}x,h_\rho^{1/2}y\>| \\
&\quad\longrightarrow\ 0\quad\mbox{as $n\to\infty$}.
\end{align*}
Since $e:=s(\eta)$ is the projection onto $\overline{M'h_\eta^{1/2}}=\overline{h_\eta^{1/2}M}$,
the above estimate implies that $eA_n^*A_ne\to A^*A$ in the weak operator topology. Therefore,
by Lemma \ref{L-5.4} one has
\begin{align}
\widehat S_f(\rho\|\eta)&=\<h_\eta^{1/2},f(A^*A)h_\eta^{1/2}\>
\le\liminf_{n\to\infty}\<h_\eta^{1/2},f(eA_n^*A_ne)h_\eta^{1/2}\> \nonumber\\
&\le\liminf_{n\to\infty}\<h_\eta^{1/2},f(A_n^*A_n)h_\eta^{1/2}\>, \label{F-5.7}
\end{align}
where the last inequality follows from $f(eT_ne)\le ef(T_n)e$ similarly to \eqref{F-5.4}.
Moreover, since $\sup_{n\ge1}\|f(A_n^*A_n)\|<+\infty$, it follows as in \eqref{F-5.6} that
$$
|\<h_{\eta_n}^{1/2},f(A_n^*A_n)h_{\eta_n}^{1/2}\>
-\<h_\eta^{1/2},f(A_n^*A_n)h_\eta^{1/2}\>|\ \longrightarrow\ 0
\quad\mbox{as $n\to\infty$},
$$
which implies that
\begin{align}\label{F-5.8}
\liminf_{n\to\infty}\widehat S_f(\rho_n\|\eta_n)
=\liminf_{n\to\infty}\<h_{\eta_n}^{1/2},f(A_n^*A_n)h_{\eta_n}^{1/2}\>
=\liminf_{n\to\infty}\<h_\eta^{1/2},f(A_n^*A_n)h_\eta^{1/2}\>.
\end{align}
Hence $\widehat S_f(\rho\|\eta)\le\liminf_{n\to\infty}\widehat S_f(\rho_n\|\eta_n)$ follows
from \eqref{F-5.7} and \eqref{F-5.8}.
\end{proof}

In the rest of this section we establish the martingale convergence for
$\widehat S_f(\rho\|\sigma)$. Let $\rho,\sigma\in M_*^+$ and $\eta:=\rho+\sigma$. Below it
will be convenient to work with the \emph{cyclic representation} $(\cH_\eta,\pi_\eta,\xi_\eta)$
of $M$ associated with $\eta$, rather than the standard representation of $M$, that is,
$\pi_\eta$ is a representation of $M$ on a Hilbert space $\cH_\eta$ with a cyclic vector
$\xi_\eta$ for $\pi_\eta(M)$, i.e., $\cH_\eta=\overline{\pi_\eta(M)\xi_\eta}$, such that
$\eta(x)=\<\xi_\eta,\pi_\eta(x)\xi_\eta\>$, $x\in M$. Then there exists a unique
$T=T(\rho/\eta)\in\pi_\eta(M)'_+$ such that
\begin{align}\label{F-5.9}
\rho(x)=\<\xi_\eta,T\pi_\eta(x)\xi_\eta\>,\qquad x\in M.
\end{align}
See, e.g., \cite[Theorems 2.3.16, 2.3.19]{BR}. Note that $T$ is independent of the choice of
the cyclic representation up to unitary conjugation. That is, let $(\hat\cH,\hat\pi,\hat\xi)$
be another cyclic representation of $M$ associated with $\eta$. There is a unitary
$U:\cH_\eta\to\hat\cH$ such that $\hat\xi=U\xi_\eta$ and $\hat\pi(x)=U\pi_\eta(x)U^*$,
$x\in M$; then $\hat T\in\hat\pi(M)'_+$ as in \eqref{F-5.9} for $(\hat\cH,\hat\pi,\hat\xi)$ is
$\hat T=UT(\rho/\eta)U^*$. A particular choice of $(\cH_\eta,\pi_\eta,\xi_\eta)$ is taken as
$\cH_\eta:=\overline{Mh_\eta^{1/2}}=L^2(M)s(\eta)$, $\pi_\eta(x)$ is the left multiplication of
$x$ on $\cH_\eta\subset L^2(M)$, and $\xi_\eta=h_\eta^{1/2}$. In this case, $T=T(\rho/\eta)$ in
\eqref{F-5.9} coincides with $T=T_{\rho/\eta}|_{\cH\eta}$, where $T_{\rho/\eta}\in M'_+$ is
given in Lemma \ref{L-3.1}. Since $h_\eta^{1/2}\in\cH_\eta$, the formula in \eqref{F-4.9} holds
as well when $T_{\rho/\eta}$ is replaced with $T_{\rho/\eta}|_{\cH_\eta}$. Therefore,
\eqref{F-4.9} is rewritten as
\begin{align}\label{F-5.10}
\widehat S_f(\rho\|\sigma)=\int_0^1(1-t)f\biggl({t\over1-t}\biggr)\,d\|E(t)\xi_\eta\|^2,
\end{align}
for any cyclic representation $(\cH_\eta,\pi_\eta,\xi_\eta)$ of $M$ associated with
$\eta=\rho+\sigma$ and the spectral decomposition $T(\rho/\eta)=\int_0^1t\,dE(t)$.

Our martingale convergence theorem is

\begin{thm}\label{T-5.6}
Let $\{M_\alpha\}$ be an increasing net of unital von Neumann subalgebras of $M$ such that
$\bigl(\bigcup_\alpha M_\alpha\bigr)''=M$. Then for every $\rho,\sigma\in M_*^+$,
$$
\widehat S_f(\rho|_{M_\alpha}\|\sigma|_{M_\alpha})\ \nearrow\ \widehat S_f(\rho\|\sigma).
$$
\end{thm}

\begin{proof}
Let $\eta:=\rho+\sigma$, $\rho_\alpha:=\rho|_{M_\alpha}$, $\sigma_\alpha:=\sigma|_{M_\alpha}$
and $\eta_\alpha:=\eta|_{M_\alpha}$. From the monotonicity property in Theorem \ref{T-2.9}
(applied to injections $M_\alpha\hookrightarrow M_\beta\hookrightarrow M$ for $\alpha\le\beta$)
we see that $\widehat S_f(\rho_\alpha\|\sigma_\alpha)$ is increasing and
$\widehat S_f(\rho_\alpha\|\sigma_\alpha)\le\widehat S_f(\rho\|\sigma)$. Hence it suffices to
show that $\widehat S_f(\rho\|\sigma)\le\sup_\alpha\widehat S_f(\rho_\alpha\|\sigma_\alpha)$.
Choose a cyclic representation $(\cH_\eta,\pi_\eta,\xi_\eta)$ of $M$ and $T=T(\rho/\eta)$
associated with $\eta=\rho+\sigma$. Let $\cH_\alpha:=\overline{\pi_\eta(M_\alpha)\xi_\eta}$
and $P_\alpha$ be the orthogonal projection from $\cH_\eta$ onto $\cH_\alpha$. Since
$$
P_\alpha\pi_\eta(x)\pi_\eta(y)\xi_\eta
=\pi_\eta(x)P_\alpha\pi_\eta(y)\xi_\eta,\qquad x,y\in M_\alpha,
$$
one has $P_\alpha\pi_\eta(x)P_\alpha=\pi_\eta(x)P_\alpha$ for any $x\in M_\alpha$, and hence
$P_\alpha\in\pi_\eta(M_\alpha)'$. So one can define a representation $\pi_\alpha$ of
$M_\alpha$ on $\cH_\alpha$ by
$$
\pi_\alpha(x):=\pi_\eta(x)|_{\cH_\alpha},\qquad x\in M_\alpha.
$$
Since $\overline{\pi_\alpha(M_\alpha)\xi_\eta}=\cH_\alpha$ and
$$
\<\xi_\eta,\pi_\alpha(x)\xi_\eta\>=\<\xi_\eta,\pi_\eta(x)\xi_\eta\>
=\eta(x)=\eta_\alpha(x),\qquad x\in M_\alpha,
$$
it follows that $(\cH_\alpha,\pi_\alpha,\xi_\eta)$ is the cyclic representation of $M_\alpha$
associated with $\eta_\alpha=\rho_\alpha+\sigma_\alpha$. Since
$$
\pi_\alpha(M_\alpha)'=(\pi_\eta(M_\alpha)|_{\cH_\alpha})'
=P_\alpha\pi_\eta(M_\alpha)'P_\alpha|_{\cH_\alpha}
\supset P_\alpha\pi_\eta(M)'P_\alpha|_{\cH_\alpha},
$$
one has $P_\alpha TP_\alpha|_{\cH_\alpha}\in\pi_\alpha(M_\alpha)_+'$. Moreover, since
$$
\<\xi_\eta,P_\alpha TP_\alpha\pi_\alpha(x)\xi_\eta\>
=\<\xi_\eta,T\pi_\eta(x)\xi_\eta\>=\rho(x)=\rho_\alpha(x),
\qquad x\in M_\alpha,
$$
we find that
\begin{align}\label{F-5.11}
T(\rho_\alpha/\eta_\alpha)=P_\alpha TP_\alpha|_{\cH_\alpha}.
\end{align}
Since $P_\alpha\nearrow I$ in the strong operator topology, it follows that
$P_\alpha TP_\alpha\to T$ in the strong operator topology.

Now, for each $n\in\bN$ let $f_n$ be given in \eqref{F-5.1} and set
$k_n(t):=(1-t)f_n\bigl({t\over1-t}\bigr)$, $t\in[0,1]$. Then $k_n$ is a continuous function
on $[0,1]$, so from formula \eqref{F-5.10} and \eqref{F-5.11} it follows that
$$
\widehat S_{f_n}(\rho\|\sigma)=\<\xi_\eta,k_n(T)\xi_\eta\>,
\qquad\widehat S_{f_n}(\rho_\alpha\|\sigma_\alpha)
=\<\xi_\eta,k_n(P_\alpha TP_\alpha)\xi_\eta\>
$$
for every $\alpha$. For each $n\in\bN$, since $k_n(P_\alpha TP_\alpha)\to k_n(T)$ in the
strong operator topology, we obtain
$$
\widehat S_{f_n}(\rho_\alpha\|\sigma_\alpha)\ \longrightarrow
\ \widehat S_{f_n}(\rho\|\sigma)\quad\mbox{as $\alpha\to``\infty"$}.
$$
From this and Lemma \ref{L-5.2} we find that
\begin{align*}
\sup_\alpha\widehat S_f(\rho_\alpha\|\sigma_\alpha)
&=\sup_\alpha\sup_n\widehat S_{f_n}(\rho_\alpha\|\sigma_\alpha)
=\sup_n\sup_\alpha\widehat S_{f_n}(\rho_\alpha\|\sigma_\alpha) \\
&\ge\sup_n\widehat S_{f_n}(\rho\|\sigma)=\widehat S_f(\rho\|\sigma),
\end{align*}
as desired.
\end{proof}

\begin{remark}\label{R-5.7}\rm
We have shown \cite[Theorem 4.1\,(i)]{Hi1} that the standard $f$-divergence $S_f(\rho\|\sigma)$
is jointly lower semicontinuous in the $\sigma(M_*,M)$-topology. It follows from Theorem
\ref{T-5.6} that this property stronger than Theorem \ref{T-5.5} holds for
$\widehat S_f(\rho\|\sigma)$ as well whenever $M$ is injective, or equivalently, there is an
increasing net $\{M_\alpha\}$ of finite-dimensional unital subalgebras of $M$ such that
$M=\bigl(\bigcup_\alpha M_\alpha\bigr)''$, see \cite{C2,Ell}. In fact, in this case,
$\widehat S_f(\rho\|\sigma)=\sup_\alpha\widehat S_f(\rho|_{M_\alpha}\|\sigma|_{M_\alpha})$
by Theorem \ref{T-5.6} and
$(\rho,\sigma)\mapsto\widehat S_f(\rho|_{M_\alpha}\|\sigma|_{M_\alpha})$ is lower
semicontinuous in the $\sigma(M_*,M)$-topology. However, it is unknown whether
$\widehat S_f(\rho\|\sigma)$ is jointly lower semicontinuous in the $\sigma(M_*,M)$-topology
for general $M$.
\end{remark}

\section{Minimal reverse test}

Let $\rho,\sigma\in M_*^+$ be arbitrary, $\eta:=\rho+\sigma$ and $e:=s(\eta)$. Let $A\in eMe$
be such that $h_\rho^{1/2}=Ah_\eta^{1/2}$ (Lemma \ref{L-2.1}). Since $0\le A^*A\le I$, we
take the spectral decomposition $A^*A=\int_0^1t\,dE(t)$, where $E(\cdot)$ is a spectral
measure on $[0,1]$ with $\int_0^1dE(t)=e$, and define a finite Borel measure $\nu$ on $[0,1]$
by
$$
\nu:=\eta(E(\cdot))=\tr(h_\eta E(\cdot))=\|E(\cdot)h_\eta^{1/2}\|^2,
$$
and consider an abelian von Neumann algebra $L^\infty([0,1],\nu)=L^1([0,1],\nu)^*$. Note that
$A^*A=JT_{\rho/\eta}J$ from the proof of Lemma \ref{L-3.3}, so that
$\nu=\|E_{\rho/\eta}(\cdot)h_\eta^{1/2}\|^2$ (see Theorem \ref{T-4.2}).

\begin{lemma}\label{L-6.1}
Define $\Phi_0:M\to L^\infty([0,1],\nu)$ by
\begin{align}\label{F-6.1}
\Phi_0(x)={d\,\tr(xh_\eta^{1/2}E(\cdot)h_\eta^{1/2})\over d\nu}
\ \mbox{(the Radon-Nikodym derivative)},\qquad x\in M.
\end{align}
Then $\Phi_0$ is a unital positive normal map and its predual map
$$
\Phi_{0*}:L^1([0,1],\nu)\ \longrightarrow\ L^1(M)\cong M_*
$$
satisfies
\begin{align}\label{F-6.2}
\Phi_{0*}(\phi)=h_\eta^{1/2}\biggl(\int_0^1\phi(t)\,dE(t)\biggr)h_\eta^{1/2}
\end{align}
for every $\phi\in L^\infty([0,1],\nu)$ ($\subset L^1([0,1],\nu)$). In particular,
$\Phi_{0*}(t)=h_\rho$ and $\Phi_{0*}(1-t)=h_\sigma$, where $t$ denotes the identity function
$t\mapsto t$ on $[0,1]$.
\end{lemma}

\begin{proof}
When $x\in M_+$, for any Borel set $S\subset[0,1]$ we have
$$
0\le\tr(xh_\eta^{1/2}E(S)h_\eta^{1/2})\le\|x\|\tr(h_\eta E(S))=\|x\|\nu(S),
$$
so that the Radon-Nikodym derivative $\Phi_0(x)$ in \eqref{F-6.1} is well defined and
$0\le\Phi_0(x)\le\|x\|1$. So $\Phi_0$ extends to a well defined positive linear map from
$M$ to $L^\infty([0,1],\nu)$. To show the normality of $\Phi_0$, let $\{x_\alpha\}$ be a
sequence in $M_+$ such that $x_\alpha\nearrow x\in M_+$. Since
$\tr(x_\alpha h_\eta^{1/2}E(\cdot)h_\eta^{1/2})$ is increasing and dominated by
$\tr(xh_\eta^{1/2}E(\cdot)h_\eta^{1/2})$, we have
$0\le\Phi_0(x_\alpha)\nearrow\psi\le\Phi_0(x)$ for some $\psi\in L^\infty([0,1],\mu)$. For
every Borel set $S\subset[0,1]$,
\begin{align*}
\int_S\psi\,d\nu&=\lim_\alpha\int_S\Phi_0(x_\alpha)\,d\nu
=\lim_\alpha\tr(x_\alpha h_\eta^{1/2}E(S)h_\eta^{1/2}) \\
&=\tr(xh_\eta^{1/2}E(S)h_\eta^{1/2})=\int_S\Phi_0(x)\,d\nu
\end{align*}
which implies that $\psi=\Phi_0(x)$ and so $\Phi_0(x_\alpha)\nearrow\Phi_0(x)$.

Hence one can take the predual map $\Phi_{0*}:L^1([0,1],\mu)\to M_*$. When
$\phi\in L^\infty([0,\infty),\mu)$
($\subset L^1([0,\infty),\mu)$), for every $x\in M$ one has
\begin{align*}
\Phi_{0*}(\phi)(x)&=\int_0^1\phi\Phi_0(x)\,d\nu
=\int_0^1\phi(t)\,d\,\tr(xh_\eta^{1/2}E(t)h_\eta^{1/2}) \\
&=\tr\biggl(xh_\eta^{1/2}\biggl(\int_0^1\phi(t)\,dE(t)\biggr)h_\eta^{1/2}\biggr),
\end{align*}
which implies \eqref{F-6.2} under the identification $M_*=L^1(M)$. Hence,
$\Phi_{0*}(1)=h_\eta^{1/2}eh_\eta^{1/2}=h_\eta$,
$\Phi_{0*}(t)=h_\eta^{1/2}(A^*A)h_\eta^{1/2}=h_\rho$ and
$\Phi_{0*}(1-t)=h_\eta-h_\rho=h_\sigma$.
\end{proof}

Now, following \cite{Ma}, we introduce the notion of reverse tests for $\rho,\sigma\in M_*^+$.

\begin{definition}\label{D-6.2}\rm
Let $(X,\cX,\mu)$ be a $\sigma$-finite measure space and $\Psi:M\to L^\infty(X,\mu)$ be a
positive linear map which is untial and normal. Then the predual map
$\Psi_*:L^1(X,\mu)\to M_*$ is trace-preserving in the sense that
$\int_X\phi\,d\mu=\Psi_*(\phi)(1)$ ($=\|\Psi_*(\phi)\|$) for every $\phi\in L^1(X,\mu)_+$.
We call a triplet $(\Psi,p,q)$ of such a map $\Psi$ and $p,q\in L^1(X,\mu)_+$ a
{\it reverse test} for $\rho,\sigma$ if $\Psi_*(p)=\rho$ and $\Psi_*(q)=\sigma$.
\end{definition}

The next variational formula of $\widehat S_f(\rho\|\sigma)$ can be the third definition of the
maximal $f$-divergences.

\begin{thm}\label{T-6.3}
For every $\rho,\sigma\in M_*^+$,
\begin{align}\label{F-6.3}
\widehat S_f(\rho\|\sigma)=\min\{S_f(p\|q):
\mbox{$(\Psi,p,q)$ a reverse test for $\rho,\sigma$}\}.
\end{align}
Moreover, $\Psi:M\to L^\infty(X,\mu)$ in \eqref{F-6.3} can be restricted to those with a
standard Borel probability space $(X,\cX,\mu)$ or more specifically to a Borel probability
space on $[0,1]$.
\end{thm}

\begin{proof}
Let $(\Psi,p,q)$ be a reverse test for $\rho,\sigma$. By the monotonicity property of
$\widehat S_f$ in Theorem \ref{T-2.9} and Example \ref{E-2.13} we have
$$
\widehat S_f(\rho\|\sigma)=\widehat S_f(\Psi(p)\|\Psi(q))
\le\widehat S_f(p\|q)=S_f(p\|q).
$$
On the other hand, $(\Phi_0,t,1-t)$ given in Lemma \ref{L-6.1} is a reverse test, for which we
have the equality $\widehat S_f(\rho\|\sigma)=S_f(t\|1-t)$. In fact, since we set
$\nu=\|E_{\rho/\eta}(\cdot)h_\eta^{1/2}\|^2$ in Lemma \ref{L-6.1}, it follows from Theorem
\ref{T-4.2} that
$$
\widehat S_f(\rho\|\sigma)=\int_0^1(1-t)f\biggl({t\over1-t}\biggr)\,d\nu(t)
=S_f(t\|1-t).
$$
Hence expression \eqref{F-6.3} follows. The latter assertion is clear from Lemma \ref{L-6.1}
(see, e.g., \cite{Sr} for standard Borel spaces).
\end{proof}

The reverse test $(\Phi_0,t,1-t)$ given in Lemma \ref{L-6.1} is a minimizer for expression
\eqref{F-6.3}, which is considered as the von Neumann algebra version of Matsumoto's
\emph{minimal} or \emph{optimal reverse test} \cite{Ma} for $\rho,\sigma$. Apply the
monotonicity property of $S_f$ \cite[Theorem 4.1\,(iv)]{Hi1} to this $\Phi_0$ (that is a unital
and completely positive normal map) to find
$$
S_f(\rho\|\sigma)\le S_f(t\|1-t)=\widehat S_f(\rho\|\sigma),
$$
so we have

\begin{thm}\label{T-6.4}
For every $\rho,\sigma\in M_*^+$,
$$
S_f(\rho\|\sigma)\le\widehat S_f(\rho\|\sigma).
$$
\end{thm}

Here is an abstract approach to quantum $f$-divergences. We say that a function
$S_f^q:M_*^+\times M_*^+\to(-\infty,+\infty]$ where $M$ varies over all von Neumann algebras
is a \emph{monotone quantum $f$-divergence} if the following are satisfied:
\begin{itemize}
\item[(a)] $S_f^q(\rho\circ\Phi\|\sigma\circ\Phi)\le S_f^q(\rho\|\sigma)$ for any unital
completely positive normal map $\Phi:M_0\to M$ between von Neumann algebras and for every
$\rho,\sigma\in M_*^+$,
\item[(b)] when $M$ is an abelian von Neumann algebra with $M=L^\infty(X,\nu)$ on a
$\sigma$-finite measure space $(X,\mu)$, $S_f^q(\rho\|\sigma)$ coincides with the classical
$f$-divergence of $\rho,\sigma\in L^1(X,\mu)_+$ as in Example \ref{E-2.13}.
\end{itemize}
If $S_f^q$ is a monotone quantum $f$-divergence, then it is clear from Theorem \ref{T-6.3}
that
$$
S_f^q(\rho\|\sigma)\le\widehat S_f(\rho\|\sigma),
$$
which justifies the name maximal $f$-divergence for $\widehat S_f$.

In the matrix case, it is easy to verify that if $\rho,\omega\in\bM_n^+$ are commuting,
then $S_f(\rho\|\omega)=\widehat S_f(\rho\|\omega)$ for every operator convex (even simply
convex) function on $(0,+\infty)$. Let us extend this to the general von Neumann algebra
setting. To do so, we first need to fix the notion of commutativity of general
$\rho,\omega\in M_*^+$. A standard way to define this is as follows: Assume that
$\omega\in M_*^+$ is faithful and let $\sigma_t^\omega$ be the modular automorphism group
associated with $\omega$. Then $\rho\in M_*^+$ is said to {\it commute} with $\omega$ if
$\rho\circ\sigma_t^\omega=\rho$ for all $t\in\bR$. Different conditions equivalent to this
were established, e.g., in terms of the Connes cocycle Radon-Nikodym derivative, in \cite{PT}
(see also \cite[\S4.10]{St}).

For (not necessarily faithful)
$\omega\in M_*^+$ with $e:=s(\omega)$ we define $\sigma_t^\omega$ ($t\in\bR$) as the modular
automorphism group on $s(\omega)Ms(\omega)$ associated with the restriction of $\omega$ to
$eMe$. The above notion of commutativity can extend to the case where $s(\rho)\le s(\omega)$,
by replacing $M$ with $eMe$ and considering $\rho,\omega$ as their restrictions to $eMe$. To
introduce the notion for general $\rho,\omega\in M_*^+$ we give the next lemma, whose proof
is deferred to Appendix B.

\begin{lemma}\label{L-6.5}
For $\rho,\omega\in M_*^+$ the following conditions are equivalent:
\begin{itemize}
\item[(i)] $\rho$ commutes with $\rho+\omega$ (i.e., $\rho\circ\sigma_t^{\rho+\omega}=\rho$
on $s(\rho+\omega)Ms(\rho+\omega)$ for all $t\in\bR$);
\item[(ii)] $\omega$ commutes with $\rho+\omega$;
\item[(iii)] $\alpha\rho+\beta\omega$ commutes with $\gamma\rho+\delta\omega$ for any
$\alpha,\beta,\gamma,\delta>0$;
\item[(iv)] $h_\rho h_\omega=h_\omega h_\rho$ (as $\tau$-measurable operators affiliated with
$N$, see the first paragraph of Section 2).
\end{itemize}

When $s(\rho)\le s(\omega)$, the above conditions are also equivalent to that $\rho$ commutes
with $\omega$ (i.e., $\rho\circ\sigma_t^\omega=\rho$ on $s(\omega)Ms(\omega)$ for all
$t\in\bR$).
\end{lemma}

\begin{definition}\label{D-6.6}\rm
For $\rho,\omega\in M_*^+$ we say that $\rho,\omega$ {\it commute} if the equivalent
conditions of Lemma \ref{L-6.5} hold. When $M$ is semifinite with a faithful normal semifinite
trace $\tau_0$, the commutativity of $\rho,\omega$ is equivalent to the commutativity of
$d\rho/d\tau_0$, $d\omega/d\tau_0\in L^1(M,\tau_0)$, see Example \ref{E-2.12}.
\end{definition}

\begin{prop}\label{P-6.7}
If $\rho,\omega\in M_*^+$ commute in the above sense, then
$$
S_f(\rho\|\omega)=\widehat S_f(\rho\|\omega)
$$
for any operator convex function $f$ on $(0,+\infty)$.
\end{prop}

\begin{proof}
By Theorem \ref{T-6.4} it suffices to prove that
$\widehat S_f(\rho\|\omega)\le S_f(\rho\|\omega)$. By \eqref{F-2.4} and
\cite[Corollary 4.4\,(3)]{Hi1} note that
\begin{align*}
\widehat S_f(\rho\|\omega)&=\lim_{\eps\searrow0}
\widehat S_f(\rho+\eps(\rho+\omega)\|\omega+\eps(\rho+\omega)), \\
S_f(\rho\|\omega)&=\lim_{\eps\searrow0}
S_f(\rho+\eps(\rho+\omega)\|\omega+\eps(\rho+\omega)).
\end{align*}
Hence, thanks to (iii) of Lemma \ref{L-6.5}, we may assume that both $\rho,\omega$ are
faithful. By assumption, $\rho$ is $\sigma_t^\omega$-invariant. Let $M_\omega$ be the
centralizer of $\omega$, i.e., $M_\omega:=\{x\in M:\sigma_t^\omega(x)=x,\,t\in\bR\}$, and
$E_\omega:M\to M_\omega$ be the conditional expectation with respect to $\omega$. Then it
follows \cite[Theorem 2.2]{HOT} that $\rho\circ E_\omega=\rho$ as well as
$\omega\circ E_\omega=\omega$. Now, since $\omega|_{M_\omega}$ is a faithful normal trace,
one can choose the Radon-Nikodym derivative $A:=d(\rho|_{M_\omega})/d(\omega|_{M_\omega})$
so that
$$
\rho(x)=\omega(Ax)=\lim_{\eps\searrow0}\omega(A(1+\eps A)^{-1}x),\qquad x\in M_\omega.
$$
Let $\cA$ be the abelian von Neumann subalgebra of $M_\omega$ generated by $A$, and
$E_\cA:M_\omega\to\cA$ be the conditional expectation with respect to $\omega|_{M_\omega}$. For
every $x\in M$ one has
\begin{align*}
\omega(x)&=\omega(E_\omega(x))=\omega(E_\cA\circ E_\omega(x)), \\
\rho(x)&=\rho(E_\omega(x))=\lim_{\eps\searrow0}\omega(A(1+\eps A)^{-1}E_\omega(x)) \\
&=\lim_{\eps\searrow0}\omega(A(1+\eps A)^{-1}E_\cA\circ E_\omega(x))
=\rho(E_\cA\circ E_\omega(x)).
\end{align*}
Therefore,
\begin{align*}
\widehat S_f(\rho\|\omega)
&=\widehat S_f(\rho\circ E_\cA\circ E_\omega\|\omega\circ E_\cA\circ E_\omega)
\le\widehat S_f(\rho|_\cA\|\omega|_\cA) \\
&=S_f(\rho|_\cA\|\omega|_\cA)\le S_f(\rho\|\omega),
\end{align*}
where the two inequalities are the monotonicity properties in Theorem \ref{T-2.9} and
\cite[Theorem 4.1\,(iv)]{Hi1}, and the second equality is due to Example \ref{E-2.13}.
\end{proof}

In particular, when $f(t)=t\log t$, we have the relation between the relative entropy $D$
and Belavkin and Staszewski's relative entropy $D_\BS$ (Example \ref{E-3.5}).

\begin{cor}\label{C-6.8}
For every $\rho,\omega\in M_*^+$,
$$
D(\rho\|\omega)\le D_\BS(\rho\|\omega),
$$
and $D(\rho\|\omega)=D_\BS(\rho\|\omega)$ if $\rho,\omega$ commute.
\end{cor}

We end this section with another martingale type convergence, which is not included in
Theorem \ref{T-5.6} since $e_\alpha Me_\alpha$'s are not unital subalgebras of $M$. Indeed,
we can prove this similarly to the proof of \cite[Theorem 4.5]{Hi1} with use of Theorem
\ref{T-5.6} and Proposition \ref{P-2.11} in view of Proposition \ref{P-6.7}.

\begin{prop}\label{P-6.9}
Let $\{e_\alpha\}$ be an increasing net of projections in $M$ such that $e_\alpha\nearrow1$.
Then for every $\rho,\omega\in M_*^+$,
$$
\lim_\alpha\widehat S_f(e_\alpha\rho e_\alpha\|e_\alpha\omega e_\alpha)
=\widehat S_f(\rho\|\omega),
$$
where $e_\alpha\omega e_\alpha$ is the restriction of $\omega$ to the reduced von Neumann
algebra $e_\alpha Me_\alpha$.
\end{prop}

\section{$C^*$-algebra case}

Let $f$ be an operator convex function on $(0,+\infty)$ as before. Let $\cA$ be a unital
$C^*$-algebra, and $\cA_+^*$ be the set of positive linear functionals (automatically bounded)
on $\cA$. In this section we extend the definition of the maximal $f$-divergence to
$\rho,\sigma\in\cA_+^*$. For any $\rho,\sigma\in\cA_+^*$ set $\eta:=\rho+\sigma$, and
$(\pi_\eta,\cH_\eta,\xi_\eta)$ be the cyclic representation of $\cA$ associated with $\eta$
so that $\eta(a)=\<\xi_\eta,\pi_\eta(a)\xi_\eta\>$, $a\in\cA$, and
$\cH_\eta=\overline{\pi_\eta(\cA)\xi_\eta}$. Then there exists a $T\in\pi_\eta(\cA)'_+$ with
$0\le T\le I$ such that
$$
\rho(a)=\<\xi_\eta,T\pi_\eta(a)\xi_\eta\>,\quad
\sigma(a)=\<\xi_\eta,(I-T)\pi_\eta(a)\xi_\eta\>,\qquad a\in\cA.
$$
The normal extensions $\tilde\rho,\tilde\sigma$ of $\rho,\sigma$ to $\pi_\eta(\cA)''$ are
defined by
$$
\tilde\rho(x):=\<\xi_\eta,Tx\xi_\eta\>,\quad
\tilde\sigma(x):=\<\xi_\eta,(I-T)x\xi_\eta\>,\qquad x\in\pi_\eta(\cA)'',
$$
so that $\rho=\tilde\rho\circ\pi_\eta$ and $\sigma=\tilde\sigma\circ\pi_\eta$.

\begin{definition}\label{D-7.1}\rm
For every $\rho,\sigma\in\cA_+^*$, with the normal extensions $\tilde\rho,\tilde\sigma$ to
$\pi_\eta(\cA)''$ ($\eta=\rho+\sigma$) we define the \emph{maximal $f$-divergence} of
$\rho,\sigma$ by
$$
\widehat S_f(\rho\|\sigma):=\widehat S_f(\tilde\rho\|\tilde\sigma).
$$
In fact, $\widehat S_f(\rho\|\sigma)$ has the same expression as \eqref{F-5.10} with the
spectral decomposition $T=\int_0^1t\,dE(t)$.
\end{definition}

\begin{lemma}\label{L-7.2}
let $\pi$ be a representation of $\cA$ on a Hilbert space $\cH$. Assume that
$\rho,\sigma\in\cA_+^*$ have the normal extensions $\overline\rho,\overline\sigma$ to
$\pi(\cA)''$, i.e., $\overline\rho,\overline\sigma$ are normal positive linear functionals on
$\pi(\cA)''$ such that $\rho=\overline\rho\circ\pi$ and $\sigma=\overline\sigma\circ\pi$. Then
$$
\widehat S_f(\rho\|\sigma)=\widehat S_f(\overline\rho\|\overline\sigma).
$$
\end{lemma}

\begin{proof}
Let $\overline\eta:=\overline\rho+\overline\sigma$,
$(\pi_{\overline\eta},\cH_{\overline\eta},\xi_{\overline\eta})$ be the cyclic representation of
$M:=\pi(A)''$ associated with $\overline\eta$, and $T\in\pi_{\overline\eta}(M)_+'$ be such that
$\overline\rho(x)=\<\xi_{\overline\eta},T\pi_{\overline\eta}(x)\xi_{\overline\eta}\>$ for all
$x\in M$. Then
$$
\eta(a)=\overline\eta(\pi(a))
=\<\xi_{\overline\eta},\pi_{\overline\eta}(\pi(a))\xi_{\overline\eta}\>,
\qquad a\in A,
$$
and
$$
\cH_{\overline\eta}=\overline{\pi_{\overline\eta}(M)\xi_{\overline\eta}}
=\overline{\pi_{\overline\eta}(\pi(\cA))\xi_{\overline\eta}}.
$$
Here, the last equality is seen
as follows: for any $x\in M$, by the Kaplansky density theorem \cite[Theorem II.4.8]{Ta1},
choose a net $a_\alpha\in\cA$ such that $\sup_\alpha\|\pi(a_\alpha)\|<+\infty$ and
$\pi(a_\alpha)\to x$ strongly*. Then
$$
\|(\pi_{\overline\eta}(\pi(a_\alpha)-x)\xi_{\overline\eta}\|^2
=\overline\eta((\pi(a_\alpha)-x)^*(\pi(a_\alpha)-x))\ \longrightarrow\ 0.
$$
Therefore, $(\pi_{\overline\eta}\circ\pi,\cH_{\overline\eta},\xi_{\overline\eta})$ is the
cyclic representation of $\cA$ associated with $\eta$. Moreover, note that
$T\in(\pi_{\overline\eta}\circ\pi)(\cA)'_+$ and
$\rho(a)=\<\xi_{\overline\eta},T(\pi_{\overline\eta}\circ\pi)(a)\xi_{\overline\eta}\>$ for all
$a\in\cA$. Hence by Definition \ref{D-7.1}, with the spectral decomposition
$T=\int_0^1t\,dE(t)$ we have
\begin{align}\label{F-7.1}
\widehat S_f(\rho\|\sigma)
=\int_0^1(1-t)f\biggl({t\over 1-t}\biggr)\,d\|E(t)\xi_{\overline\eta}\|^2.
\end{align}
On the other hand, applying \eqref{F-5.10} to $\overline\rho,\overline\sigma\in M_*^+$ shows
that $\widehat S_f(\overline\rho\|\overline\sigma)$ has the same integral expression as
\eqref{F-7.1}, so the asserted equality follows.
\end{proof}

The above lemma says that $\widehat S_f(\rho\|\sigma)$ for $\rho,\sigma\in\cA_+^*$ can be
defined as $\widehat S_f(\overline\rho\|\overline\sigma)$ via any
$(\pi,\overline\rho,\overline\sigma)$ of a representation $\pi$ of $\cA$ and normal extensions
$\overline\rho,\overline\sigma$ of $\rho,\sigma$ to $\pi(\cA)''$. An example of such
representation, besides $\pi_\eta$ in Definition \ref{D-7.1}, is the universal representation
$\pi$ of $\cA$, for which $\pi(\cA)''\cong\cA^{**}$ (isometric to the second conjugate space of
$\cA$), see \cite[\S III.2]{Ta1}.

In the rest of the section we give some basic properties $\widehat S_f(\rho\|\sigma)$ for
$\rho,\sigma\in\cA_+^*$.

\begin{prop}\label{P-7.3}
The function $(\rho,\sigma)\in\cA_+^*\times\cA_+^*\mapsto\widehat S_f(\rho\|\sigma)$ is jointly
convex and jointly lower semicontinuous in the norm topology.
\end{prop}

\begin{proof}
Let $\pi$ be the universal representation of $\cA$. For $\rho_i,\sigma_i\in M_*^+$ and
$\lambda_i\ge0$ ($1\le i\le n$) let $\overline\rho_i,\overline\sigma_i$ be the normal extensions
of $\rho_i,\sigma_i$ to $\pi(\cA)''$. By Lemma \ref{L-7.2} and the joint convexity property in
Theorem \ref{T-2.9},
$$
\widehat S_f\Biggl(\sum_{i=1}^n\lambda_i\rho_i\Bigg\|\sum_{i=1}^n\lambda\sigma_i\Biggr)
=\widehat S_f\Biggl(\sum_{i=1}^n\lambda_i\overline\rho_i\Bigg\|
\sum_{i=1}^n\lambda\overline\sigma_i\Biggr)
\le\sum_{i=1}^n\lambda_iS_f(\overline\rho_i\|\overline\sigma_i)
=\sum_{i=1}^n\lambda_iS_f(\rho_i\|\sigma_i).
$$
Next, let $\rho_n,\rho,\sigma_n,\sigma\in\cA_+^*$, $n\in\bN$, be such that $\|\rho_n-\rho\|\to0$
and $\|\sigma_n-\sigma\|\to0$. From the Kaplansky density theorem and \cite[Lemma IV.3.8]{Is},
it follows that
$$
\|\overline\rho_n-\rho\|
=\sup_{x\in\pi(\cA)'',\|x\|\le1}|(\overline\rho_n-\overline\rho)(x)|
=\sup_{a\in\cA,\|a\|\le1}|(\overline\rho_n-\overline\rho)(\pi(a))|
=\|\rho_n-\rho\|\ \longrightarrow\ 0,
$$
and similarly $\|\overline\sigma_n-\overline\sigma\|=\|\sigma_n-\sigma\|\to0$. Hence by Lemma
\ref{L-7.2} and Theorem \ref{T-5.5},
$$
\widehat S_f(\rho\|\sigma)=\widehat S_f(\overline\rho\|\overline\sigma)
\le\liminf_{n\to\infty}\widehat S_f(\overline\rho_n\|\overline\sigma_n)
=\liminf_{n\to\infty}\widehat S_f(\rho_n\|\sigma_n),
$$
showing the lower semicontinuity in the norm topology.
\end{proof}

\begin{prop}\label{P-7.4}
Let $\cA_0$ be another unital $C^*$-algebra and $\Phi:\cA_0\to\cA$ be a unital positive linear
map. Then for every $\rho,\sigma\in\cA_+^*$,
$$
\widehat S_f(\rho\circ\Phi\|\sigma\circ\Phi)\le\widehat S_f(\rho\|\sigma).
$$
\end{prop}

\begin{proof}
Let $\pi,\pi_0$ be the universal representations of $\cA,\cA_0$, respectively. One can define
the unital positive normal map $\overline\Phi:\pi_0(\cA_0)''\to\pi(\cA)''$ subject to the
commutative diagrams (see \cite[Lemma III.2.2]{Ta1}):
$$
\begin{CD}
\cA_0 @> i >> \cA_0^{**} @> \cong >> \pi_0(\cA_0)''\\
@V \Phi VV  @VV \Phi^{**} V @VV \overline\Phi V\\
\cA @> i >> \cA^{**} @> \cong >> \pi(\cA)''
\end{CD},
$$
where $i$ is the canonical imbedding of $\cA$ ($\cA_0$) into $\cA^{**}$ ($\cA_0^{**}$) and
$\Phi^{**}$ is the second conjugate map. Here it is immediate to verify that
$\overline\Phi\circ\pi_0(a)=\pi\circ\Phi(a)$ for all $a\in\cA_0$. Moreover, the positivity of
$\overline\Phi$ is seen as follows: for any $x\in\pi_0(\cA_0)''$, by the Kaplansky density
theorem, choose a net $a_\alpha\in\cA_0$ such that
$\sup_\alpha\|\pi_0(a_\alpha)\|<+\infty$ and $\pi_0(a_\alpha)\to x$ strongly*. Then
$\pi_0(a_\alpha^*a_\alpha)\to x^*x$ strongly so that
$\overline\Phi(\pi_0(a_\alpha^*a_\alpha))\to\overline\Phi(x^*x)$ weakly. Since
$\overline\Phi(\pi_0(a_\alpha^*a_\alpha))=\pi(\Phi(a_\alpha^*a_\alpha))\ge0$,
$\overline\Phi(x^*x)\ge0$. For every $\rho\in\cA_+^*$, since $\overline\rho\circ\overline\Phi$
is normal on $\pi_0(\cA_0)''$ and
$$
\overline\rho\circ\overline\Phi\circ\pi_0(a)
=\overline\rho\circ\pi\circ\Phi(a)=\rho\circ\Phi(a),\qquad a\in\cA_0,
$$
one sees that $\overline\rho\circ\overline\Phi$ is the normal extension of $\rho\circ\Phi$ to
$\pi_0(\cA_0)''$. Therefore, by Lemma \ref{L-7.2} and the monotonicity property in Theorem
\ref{T-2.9},
$$
\widehat S_f(\rho\circ\Phi\|\sigma\circ\Phi)
=\widehat S_f(\overline\rho\circ\overline\Phi\|\overline\sigma\circ\overline\Phi)
\le\widehat S_f(\overline\rho\|\overline\sigma)=\widehat S_f(\rho\|\sigma)
$$
for all $\rho,\sigma\in\cA_+^*$.
\end{proof}

\begin{prop}\label{P-7.5}
Let $\{\cA_\alpha\}$ be an increasing net of unital $C^*$-subalgebras of $\cA$ such that
$\bigcup_\alpha\cA_\alpha$ is norm-dense in $\cA$. Then for every $\rho,\sigma\in\cA_+^*$,
$$
\widehat S_f(\rho|_{\cA_\alpha}\|\sigma|_{\cA_\alpha})
\ \nearrow\ \widehat S_f(\rho\|\sigma).
$$
\end{prop}

\begin{proof}
With the universal representation $\pi$ of $\cA$ we have an increasing net
$\{\pi(\cA_\alpha)''\}$ of unital von Neumann subalgebras of $\pi(\cA)''$ such that
$\bigl(\bigcup_\alpha\pi(\cA_\alpha)\bigr)''=\pi(\cA)''$. Hence by Lemma \ref{L-7.2} and
Theorem \ref{T-5.6},
$$
\widehat S_f(\rho|_{\cA_\alpha}\|\sigma|_{\cA_\alpha})
=\widehat S_f\bigl(\overline\rho|_{\pi(\cA_\alpha)''}\|
\overline\sigma|_{\pi(\cA_\alpha)''}\bigr)
\ \nearrow\ \widehat S_f(\overline\rho\|\overline\sigma)=\widehat S_f(\rho\|\sigma)
$$
for all $\rho,\sigma\in\cA_+^*$.
\end{proof}

\section{Closing remarks and problems}

In the previous paper \cite{Hi1} we discussed standard $f$-divergences $S_f(\rho\|\sigma)$ in
von Neumann algebra setting for general operator convex functions $f$ on $(0,+\infty)$. In this
paper we present a systematic study of another type of quantum $f$-divergences
$\widehat S_f(\rho\|\sigma)$ called the maximal $f$-divergences in the same setting. Starting
with a rather abstract definition (Definition \ref{D-2.8}) we present more explicit expressions
of $\widehat S_f(\rho\|\sigma)$ in an integral formula (Theorem \ref{T-4.2}) and in a
variational formula (Theorem \ref{T-6.3}), from which we can derive several important
properties of $\widehat S_f(\rho\|\sigma)$. Properties of $S_f(\rho\|\sigma)$ and
$\widehat S_f(\rho\|\sigma)$ are common in most cases, but there are also small differences
between them. For instance, the joint lower semicontinuity of $S_f(\rho\|\sigma)$ holds in the
$\sigma(M_*,M)$-topology, but that of $\widehat S_f(\rho\|\sigma)$ is shown in the norm
topology, and it is open whether $\widehat S_f$ has the same property in the
$\sigma(M_*,M)$-topology, as mentioned in Remark \ref{R-5.7}. The monotonicity inequality (DPI)
of $S_f$ holds under unital Schwarz normal maps, while that of $\widehat S_f$ is shown more
generally under unital simply positive normal maps.

We have the general inequality $S_f\le\widehat S_f$ (Theorem \ref{T-6.4}). For matrices
$\rho,\sigma\in\bM_d^+$ with $s(\rho)\le s(\sigma)$, it was shown in \cite[Theorem 4.3]{HM} that
$S_f(\rho\|\sigma)=\widehat S_f(\rho\|\sigma)$ holds if and only if $\rho\sigma=\sigma\rho$,
under a mild assumption on the support of the representing measure for the integral expression
of $f$. In particular, for matrices $\rho,\sigma$ with $s(\rho)\le s(\sigma)$,
$D(\rho\|\sigma)=D_\BS(\rho\|\sigma)$ holds if and only if $\rho\sigma=\sigma\rho$. An
interesting problem is to extend this result to the von Neumann algebra setting. In the proof
of \cite[Theorem 4.3]{HM} we used the reversibility via equality in the monotonicity inequality
for $S_f$. Thus, our next research topic should be the reversibility question under equality in
the monotonicity inequality for $S_f$. Here we say that a unital normal map $\Phi:M_0\to M$
(which satisfies a kind of positivity such as complete positivity) is \emph{reversible} for
$\{\rho,\sigma\}$ in $M_*^+$ if there exists a map $\Psi:M\to M_0$ of similar kind such that
$\rho\circ\Phi\circ\Psi=\rho$ and $\sigma\circ\Phi\circ\Psi=\sigma$. The question says whether
$\Phi$ is reversible for $\{\rho,\sigma\}$ or not if
$S_f(\rho\circ\Phi\|\sigma\circ\Phi)=S_f(\rho\|\sigma)<+\infty$. In the matrix case, the
question was well studied in \cite{HMPB,HM}, including discussions on the equality case in the
monotonicity inequality for $\widehat S_f$. For reversibility in the von Neumann algebra case,
former results in some special cases of relative entropy and the standard R\'enyi divergences
are found in, e.g., \cite{Pe1,Pe2,JP}, and recent results in the case of sandwiched R\'enyi
divergences are obtained in \cite{Je1,Je2}.

The notion of quantum $f$-divergences in the opposite direction to $\widehat S_f$ is that of
measured (or minimal) $f$-divergences, whose matrix case was discussed in \cite{HM}. For
$\rho,\sigma\in M_*^+$, taking account of Theorem \ref{T-6.3} one can define the
\emph{measured $f$-divergence} $S_f^\meas(\rho\|\sigma)$ of $\rho$ with respect to $\sigma$ by
$$
S_f^\meas(\rho\|\sigma):=\sup\{S_f(\rho\circ\Phi\|\sigma\circ\Phi):
\Phi:L^\infty([0,1],\nu)\to M\},
$$
where $\Phi$ runs over unital positive normal map from $L^\infty([0,1],\nu)$ on a Borel
probability space $([0,1],\nu)$ to $M$. When $M$ is $\sigma$-finite so that it has a faithful
normal state, note that the above $S_f^\meas(\rho\|\sigma)$ is the supremum of the classical
$f$-divergence $S_f(\cM(\rho)\|\cM(\sigma)$ over $M$-valued measurements $\cM$ on $[0,1]$, i.e.,
$\cM$ is $\sigma$-additive $M_+$-valued measure with $\cM([0,1])=1$, where $\cM(\rho)$ denotes
a Borel measure $\rho(\cM(\cdot))$ on $[0,1]$. One can also define $S_f^\pr(\rho\|\sigma)$ by
restricting $\cM$ in the above to measurements whose values are projections in $M$. Due to the
monotonicity of $S_f$ in \cite[Theorem 4.1\,(iv)]{Hi1} it is clear that
$$
S_f^\pr(\rho\|\sigma)\le S_f^\meas(\rho\|\sigma)\le S_f(\rho\|\sigma)
\le\widehat S_f(\rho\|\sigma).
$$
It is interesting to characterize the equality case $S_f(\rho\|\sigma)=S_f^\meas(\rho\|\sigma)$
or $S_f(\rho\|\sigma)=S_f^\pr(\rho\|\sigma)$ in terms of commutativity of $\rho,\sigma$, as in
\cite[Theorem 4.18]{HM} in the matrix case.

Apart from the conventional (or standard) R\'enyi divergences, a new type of R\'enyi divergences
called the sandwiched ones have extensively been developed in these years. For matrices
$\rho,\sigma\in\bM_d^+$, the \emph{sandwiched R\'enyi divergence}
$\widetilde D_\alpha(\rho\|\sigma)$ for $\alpha\in(0,\infty)\setminus\{1\}$ is defined by
$$
\widetilde D_\alpha(\rho\|\sigma)
:={1\over\alpha-1}\log{\widetilde Q_\alpha(\rho\|\sigma)\over\Tr\rho},\qquad
\widetilde Q_\alpha(\rho\|\sigma)
:=\Tr\bigl(\sigma^{1-\alpha\over2\alpha}\rho\sigma^{1-\alpha\over2\alpha}\bigr)^\alpha,
$$
while the standard R\'enyi divergence is
$$
D_\alpha(\rho\|\sigma):={1\over\alpha-1}\log{\Tr\rho^\alpha\sigma^{1-\alpha}\over\Tr\rho},
$$
where $\sigma^\gamma$ for $\gamma<0$ is defined via the generalized inverse. It is widely known
\cite{HiPe,ON,MO} that $D_\alpha(\rho\|\sigma)$ and $\widetilde D_\alpha(\rho\|\sigma)$,
together with $D_1(\rho\|\sigma)=D(\rho\|\sigma)/\Tr\rho$, play a significant role in quantum
state discrimination, thus enjoying good operational interpretation. The extension of
$\widetilde D_\alpha(\rho\|\sigma)$ to the von Neumann algebra setting has been made in recent
papers \cite{BST,Je1,Je2}, and a detailed exposition on $D_\alpha(\rho\|\sigma)$ in von Neumann
algebras has been given in \cite{Hi1}. For $\rho,\sigma\in M_*^+$, the
\emph{maximal R\'enyi divergence} $\widehat D_\alpha(\rho\|\sigma)$ is defined by
$$
\widehat D_\alpha(\rho\|\sigma)
:={1\over\alpha-1}\log{\widehat Q_\alpha(\rho\|\sigma)\over\rho(1)},\qquad
\widehat Q_\alpha(\rho\|\sigma):=\begin{cases}
\widehat S_{f_\alpha}(\rho\|\sigma) & \text{if $\alpha>1$}, \\
-\widehat S_{f_\alpha}(\rho\|\sigma) & \text{if $0<\alpha<1$},
\end{cases}
$$
where $f_\alpha(t):=t^\alpha$ ($\alpha>1$) and $-t^\alpha$ ($0<\alpha<1$). (Although $f_\alpha$
for $\alpha>2$ is not operator convex on $(0,+\infty)$, one can define
$\widehat S_{f_\alpha}(\rho\|\sigma)$, for instance, by the integral expression in
\eqref{F-4.9} with $f=f_\alpha$.) For matrices $\rho,\sigma\in\bM_d^+$ with
$s(\rho)\le s(\sigma)$, we have
$\widehat Q_\alpha(\rho\|\sigma)=\Tr\sigma(\sigma^{-1/2}\rho\sigma^{-1/2})^\alpha$ by
\eqref{F-4.14}, and from \cite[Remark 4.6]{HM} we see that
\begin{align*}
&\widetilde D_\alpha(\rho\|\sigma)\le D_\alpha(\rho\|\sigma)
\le\widehat D_\alpha(\rho\|\sigma)\quad\mbox{for $\alpha\in(0,2]\setminus\{1\}$}, \\
&\widetilde D_\alpha(\rho\|\sigma)\le\widehat D_\alpha(\rho\|\sigma)
\le D_\alpha(\rho\|\sigma)\quad\mbox{for $\alpha\in[2,\infty)$}.
\end{align*}
In the von Neumann algebra setting, it follows from Theorem \ref{T-6.4} that
$D_\alpha\le\widehat D_\alpha$ for $\alpha\in(0,2)\setminus\{1\}$, while it was shown in
\cite{BST,Je1} that $\widetilde D_\alpha\le D_\alpha$ for $\alpha\in[1/2,\infty)\setminus\{1\}$.
But comparison between between $D_\alpha$, $\widetilde D_\alpha$ and $\widehat D_\alpha$ in
the von Neumann algebra case has not fully been investigated.

\subsection*{ Acknowledgments}

This work was supported by JSPS KAKENHI Grant Number JP17K05266.

\appendix

\section{Proofs of \eqref{F-4.6} and \eqref{F-4.7}}

\subsection{Proof of \eqref{F-4.7}}

We use the integral expression of $f$ in \eqref{F-3.9}. Let $\eps\in(0,1/2)$ and we divide
the integral on $[0,1]$ into two parts on $[0,2/3]$ and
$(2/3,1]$. Since $1-t+\eps t\ge1/3$ for $\eps\in(0,1/2)$ and $t\in[0,2/3]$,
$(1-t+\eps t)f\bigl({t\over1-t+\eps t}\bigr)$ is uniformly bounded for those $\eps,t$ and
converges to $(1-t)f\bigl({t\over1-t}\bigr)$ as $1/2>\eps\searrow0$. Hence the bounded
convergence theorem gives
$$
\lim_{\eps\searrow0}\int_{[0,2/3]}
(1-t+\eps t)f\biggl({t\over1-t+\eps t}\biggr)\,d\nu(t)
=\int_{[0,2/3]}(1-t)f\biggl({t\over1-t}\biggr)\,d\nu(t).
$$
Next, we write
\begin{align*}
(1-t+\eps t)f\biggl({t\over1-t+\eps t}\biggr)
&=a+(b-a+\eps)t+c\,{t^2\over1-t+\eps t} \\
&\qquad+\int_{(0,+\infty)}(1-t+\eps t)
\psi_s\biggl({t\over1-t+\eps t}\biggr)\,d\mu(s).
\end{align*}
Note that for every $t\in(2/3,1]$,
$$
0<{t^2\over1-t+\eps t}\ \ \nearrow\ \ {t^2\over1-t}\quad\mbox{as $\eps\searrow0$},
$$
where ${t^2\over1-t}=+\infty$ for $t=1$. Also, when $s\in(0,+\infty)$ and $t\in(2/3,1]$,
since
$$
(1-t+\eps t)\psi_s\biggl({t\over1-t+\eps t}\biggr)
={t\over1+s}-{t\over{t\over1-t+\eps t}+s}
$$
and ${t\over1-t+\eps t}>1$ for every $\eps\in(0,1/2)$, we note that
$$
0<(1-t+\eps t)\psi_s\biggl({t\over1-t+\eps t}\biggr)
\ \ \nearrow\ \ (1-t)\psi_s\biggl({t\over1-t}\biggr)\quad\mbox{as $1/2>\eps\searrow0$},
$$
where $(1-t)\psi_s\bigl({t\over1-t}\bigr)={1\over1+s}$ for $t=1$. By the monotone convergence
theorem we find that
\begin{align*}
&\int_{(2/3,1]}(1-t+\eps t)f\biggl({t\over1-t+\eps t}\biggr)\,d\nu(t) \\
&\quad=a\int_{(2/3,1]}d\nu(t)+(b-a+\eps)\int_{(2/3,1]}t\,d\nu(t)
+c\,\int_{(2/3,1]}{t^2\over1-t+\eps t}\,d\nu(t) \\
&\qquad\quad+\int_{(2/3,1]}\int_{(0,+\infty)}(1-t+\eps t)
\psi_s\biggl({t\over1-t+\eps t}\biggr)\,d\mu(s)\,d\nu(t)
\end{align*}
converges as $1/2>\eps\searrow0$ to
\begin{align*}
&a\int_{(2/3,1]}d\nu(t)+(b-a)\int_{(2/3,1]}t\,d\nu(t)
+c\,\int_{(2/3,1]}{t^2\over1-t}\,d\nu(t) \\
&\qquad+\int_{(2/3,1]}\int_{(0,+\infty)}(1-t)
\psi_s\biggl({t\over1-t}\biggr)\,d\mu(s)\,d\nu(t) \\
&\quad=\int_{(2/3,1]}\biggl(a+(b-a)t+c\,{t^2\over1-t}
+\int_{(0,+\infty)}(1-t)\psi_s\biggl({t\over1-t}\biggr)\,d\mu(s)\biggr)\,d\nu(t) \\
&\quad=\int_{(2/3,1]}(1-t)f\biggl({t\over1-t}\biggr)\,d\nu(t).
\end{align*}
Therefore, \eqref{F-4.7} follows.\qed

\subsection{Proof of \eqref{F-4.6}}
Let $\eps\in(0,1)$ and we divide the integral on $[0,1]$ into two parts on $[0,2/3]$ and
$(2/3,1]$. As in the proof of \eqref{F-4.7} we have
$$
\lim_{\eps\searrow0}\int_{[0,2/3]}(1-t+\eps t)
f\biggl({t+\eps(1-t)\over1-t+\eps t}\biggr)\,d\nu(t)
=\int_{[0,2/3]}(1-t)f\biggl({t\over1-t}\biggr)\,d\nu(t).
$$
Next, we write
\begin{align*}
(1-t+\eps t)f\biggl({t+\eps(1-t)\over1-t+\eps t}\biggr)
&=(a+\eps b)+(b-a)(1-\eps)t+c\,{(t+\eps(1-t))^2\over1-t+\eps t} \\
&\qquad+\int_{(0,+\infty)}(1-t+\eps t)
\psi_s\biggl({t+\eps(1-t)\over1-t+\eps t}\biggr)\,d\mu(s).
\end{align*}
Compute
$$
{d\over d\eps}\,{(t+\eps(1-t))^2\over1-t+\eps t}
={(t+\eps(1-t))(2-4t+t^2+\eps t(1-t))\over(1-t+\eps t)^2}.
$$
When $\eps\in(0,1)$ and $t\in(2/3,1]$, since
$$
2-4t+t^2+\eps t(1-t)\le2-4t+t^2+t(1-t)=2-3t<0,
$$
we have
$$
{d\over d\eps}\,{(t+\eps(1-t))^2\over1-t+\eps t}<0,
$$
so that
$$
{(t+\eps(1-t))^2\over1-t+\eps t}
\ \ \nearrow\ \ {t^2\over1-t}\quad\mbox{as $\eps\searrow0$}.
$$
Also, when $s\in(0,+\infty)$ and $t\in(2/3,1]$, since
$$
(1-t+\eps t)\psi_s\biggl({t+\eps(1-t)
\over1-t+\eps t}\biggr)
={t+\eps(1-t)\over1+s}-{t+\eps(1-t)\over
{t+\eps(1-t)\over1-t+\eps t}+s}
$$
and for every $\eps\in(0,1)$, ${t\over1-t+\eps t}>1$ and
$$
{d\over d\eps}\,{t+\eps(1-t)\over1-t+\eps t}={1-2t\over(1-t+\eps t)^2}<0,
$$
we note that
$$
0<(1-t+\eps t)\psi_s\biggl({t+\eps(1-t)\over
1-t+\eps t}\biggr)
\ \ \nearrow\ \ (1-t)\psi_s\biggl({t\over1-t}\biggr)\quad\mbox{as $1>\eps\searrow0$}.
$$
By the monotone convergence theorem we find that
\begin{align*}
&\int_{(2/3,1]}(1-t+\eps t)
f\biggl({t+\eps(1-t)\over1-t+\eps t}\biggr)\,d\nu(t) \\
&\quad=(a+\eps b)\int_{(2/3,1]}d\nu(t)+(b-a)(1-\eps)\int_{(2/3,1]}t\,d\nu(t) \\
&\qquad\quad+c\,\int_{(2/3,1]}{(t+\eps(1-t))^2\over1-t+\eps t}\,d\nu(t) \\
&\qquad\quad+\int_{(2/3,1]}\int_{(0,+\infty)}(1-t+\eps t)
\psi_s\biggl({t+\eps(1-t)\over1-t+\eps t}\biggr)\,d\mu(s)\,d\nu(t)
\end{align*}
converges as $1>\eps\searrow0$ to
\begin{align*}
&a\int_{(2/3,1]}d\nu(t)+(b-a)\int_{(2/3,1]}t\,d\nu(t)
+c\,\int_{(2/3,1]}{t^2\over1-t}\,d\nu(t) \\
&\qquad+\int_{(2/3,1]}\int_{(0,+\infty)}(1-t)
\psi_s\biggl({t\over1-t}\biggr)\,d\mu(s)\,d\nu(t) \\
&\quad=\int_{(2/3,1]}\biggl(a+(b-a)t+c\,{t^2\over1-t}
+\int_{(0,+\infty)}(1-t)\psi_s\biggl({t\over1-t}\biggr)\,d\mu(s)
\biggr)\,d\nu(t) \\
&\quad=\int_{(2/3,1]}(1-t)f\biggl({t\over1-t}\biggr)\,d\nu(t).
\end{align*}
Therefore, \eqref{F-4.6} follows.\qed

\section{Proof of Lemma \ref{L-6.5}}

For the proof below, we may assume that $\rho+\omega$ is faithful, by replacing $M$ with
$eMe$, where $e:=s(\rho+\omega)$. Here, concerning (iv), we note that $(eMe)_*$ is identified
with $eM_*e$, where $e\psi e(x):=\psi(exe)$, $x\in M$, for $\psi\in M_*^+$, and by
\cite[Theorem 7]{Te}, $h_{e\psi e}=eh_\psi e$ for every $\psi\in M_*$ so that
$$
(eMe)_*\cong eL^1(M)e=\{eh_\psi e:\psi\in M_*\}.
$$

(i)$\iff$(ii).\enspace
Since $\rho+\omega$ is invariant under $\sigma_t^{\rho+\omega}$, (i) implies that
$$
\omega\circ\sigma_t^{\rho+\omega}
=(\rho+\omega)\circ\sigma_t^{\rho+\omega}-\rho\circ\sigma_t^{\rho+\omega}
=(\rho+\omega)-\rho=\omega.
$$
Hence (i)$\implies$(ii), and the converse is similar.

(i)$\iff$(iii).\enspace
Assume (i) (hence also (ii)), and let $\gamma,\delta>0$ be arbitrary. Since $\rho,\omega$ are
invariant under $\sigma_t^{\rho+\omega}$,
$(\gamma\rho+\delta\omega)\circ\sigma_t^{\rho+\omega}=\gamma\rho+\delta\omega$. Hence by
\cite{PT},
$$
(\rho+\omega)\circ\sigma_t^{\gamma\rho+\delta\omega}=\rho+\omega.
$$
Assume that $\gamma\ne\delta$. Since
$(\gamma\rho+\delta\omega)\circ\sigma_t^{\gamma\rho+\delta\omega}=\gamma\rho+\delta\omega$, it
follows that $\rho,\omega$ are invariant under $\sigma_t^{\gamma\rho+\delta\omega}$ and so is
$\alpha\rho+\beta\gamma$ for any $\alpha,\beta>0$. When $\gamma=\delta$, the same follows from
(i) and (ii) immediately. Hence we have (i)$\implies$(iii). The converse is easy.

(i)$\implies$(iv).\enspace
Assume (i) and hence (iii). Let $\alpha,\beta,\gamma,\delta>0$. Since (see \cite[(18)]{K1})
$$
\sigma_t^{\gamma\rho+\delta\omega}(x)
=h_{\gamma\rho+\delta\omega}^{it}xh_{\gamma\rho+\delta\omega}^{-it},\qquad x\in M,
$$
it follows that
$$
\tr(h_{\alpha\rho+\beta\omega}h_{\gamma\rho+\delta\omega}^{it}x
h_{\gamma\rho+\delta\omega}^{-it})
=\tr(h_{\alpha\rho+\beta\omega}x),\qquad x\in M,\ t\in\bR,
$$
so that
$$
h_{\gamma\rho+\delta\omega}^{-it}h_{\alpha\rho+\beta\omega}
h_{\gamma\rho+\delta\omega}^{it}
=h_{\alpha\rho+\beta\omega},\qquad t\in\bR,
$$
as elements of $L^1(M)$ and hence as elements in $\widetilde N$. This implies that
$h_{\alpha\rho+\beta\omega}$ commutes with $h_{\gamma\rho+\delta\omega}$ in the sense of
\cite[p.\ 271]{RS}. Since those are $\tau$-measurable operators, we have
$$
h_{\alpha\rho+\beta\omega}h_{\gamma\rho+\delta\omega}
=h_{\gamma\rho+\delta\omega}h_{\alpha\rho+\beta\omega}
$$
so that
$$
(\alpha h_\rho+\beta h_\omega)(\gamma h_\rho+\delta h_\omega)
=(\gamma h_\rho+\delta h_\omega)(\alpha h_\rho+\beta h_\omega).
$$
Choosing $\alpha=\delta=1$ and $\beta,\gamma\to0$ gives (iv).

(iv)$\implies$(i).\enspace
Since (iv) implies that $h_\rho h_{\rho+\omega}=h_{\rho+\omega}h_\rho$, it is easy to verify
that $h_\rho$ commutes with $(h_{\rho+\omega}+1)^{-1}$. This implies that $h_\rho$ commutes
with any spectral projection of $h_{\rho+\omega}$. Therefore,
$h_\rho h_{\rho+\omega}^{it}=h_{\rho+\omega}^{it}h_\rho$, $t\in\bR$, from which we have
$$
\tr(h_\rho h_{\rho+\omega}^{it}xh_{\rho+\omega}^{-it})=\tr(h_\rho x),\qquad
x\in M,\ t\in\bR.
$$
Hence (i) follows.

Next, when $s(\rho)\le s(\omega)$ (hence we may assume that $\omega$ is faithful), it is seen
as above (by replacing $\rho+\omega$ with $\omega$) that (iv) is equivalent to that
$h_\rho h_\omega^{it}=h_\omega^{it}h_\rho$ for any $t\in\bR$, which means that
$\rho\circ\sigma_t^\omega=\rho$, $t\in\bR$.\qed

\end{document}